\documentclass[letterpaper, 12pt]{article}
\usepackage{amsfonts, amsmath, amsthm, graphicx, float, subcaption, enumerate, color, natbib, algorithm, url}
\usepackage[top=1in,bottom=1in,left=1in,right=1in]{geometry}

\graphicspath{{figures/}}

\newtheorem{assumption}{Assumption} \newtheorem{lemma}{Lemma}
\newtheorem{theorem}{Theorem} \theoremstyle{remark}
\newtheorem*{remark}{{\bf Remark}} \newtheorem{example}{Example}

\allowdisplaybreaks

\begin{document}
\title{Optimal Subsampling for Large Sample Logistic Regression}

\author{HaiYing Wang
  \footnote{Department of Statistics,
    University of Connecticut, Storrs, CT 06269},
  Rong Zhu
  \footnote{Academy of Mathematics and Systems Science, Chinese Academy of Sciences, Beijing, China},
  Ping Ma
  \footnote{Department of Statistics, University of Georgia, Athens, GA 30602}}
\date{}
\maketitle
\begin{abstract}
  For massive data, the family of subsampling algorithms is popular to
  downsize the data volume and reduce computational burden. Existing
  studies focus on approximating the ordinary least squares estimate
  in linear regression, where statistical leverage scores are often
  used to define subsampling probabilities. In this paper, we propose
  fast subsampling algorithms to efficiently approximate the maximum
  likelihood estimate in logistic regression.  We first establish
  consistency and asymptotic normality of the estimator from a general
  subsampling algorithm, and then derive optimal subsampling
  probabilities that minimize the asymptotic mean squared error of the
  resultant estimator. An alternative minimization criterion is also
  proposed to further reduce the computational cost. The optimal
  subsampling probabilities depend on the full data estimate, so we
  develop a two-step algorithm to approximate the optimal subsampling
  procedure. This algorithm is computationally efficient and has a
  significant reduction in computing time compared to the full data
  approach. Consistency and asymptotic normality of the estimator from
  a two-step algorithm are also established.  Synthetic and real data
  sets are used to evaluate the practical performance of the proposed
  method.
\end{abstract}

\noindent%
  \vfill

\newpage
\section{Introduction}
\label{sec:introduction}
With the rapid development of science and technologies, massive data
have been generated at an extraordinary speed.  Unprecedented volumes
of data offer researchers both unprecedented opportunities and
challenges. The key challenge is that directly applying statistical
methods to these super-large sample data using conventional computing
methods is prohibitive. We shall now present two motivating examples.

\begin{example}\label{PExp1}{\it Census}.
  The U.S. census systematically acquires and records data of all
  residents of the United States.  The census data provide fundamental
  information to study socio-economic issues.  \cite{kohavi-nbtree}
  conducted a classification analysis using residents' information
  such as income, age, work class, education, the number of working
  hours per week, and etc.  They used these information to predict
  whether the residents are high income residents, i.e., those with
  annal income more than $\$50$K, or not.  Given that the whole census
  data is super-large, the computation of statistical analysis is very
  difficult.
\end{example}

\begin{example}\label{PExp2}{\it Supersymmetric Particles}.
  Physical experiments to create exotic particles that occur only at
  extremely high energy densities have been carried out using modern
  accelerators. e.g., large Hadron Collider (LHC).  Observations of
  these particles and measurements of their properties may yield
  critical insights about the fundamental properties of the physical
  universe.  One particular example of such exotic particles is
  supersymmetric particles, the search of which is a central
  scientific mission of the LHC \citep{Baldi2014}.  Statistical
  analysis is crucial to distinguish collision events which produce
  supersymmetric particles (signal) from those producing other
  particles (background).  Since LHC continuously generates petabytes
  of data each year, the computation of statistical analysis is very
  challenging.
\end{example}

The above motivating examples are classification problems with massive
data.  Logistic regression models are widely used for classification
in many disciplines, including business, computer science, education,
and genetics, among others \citep{hosmer2013applied}.  Given
covariates ${\mathbf{x}}_i$'s $\in\mathbb{R}^d$, logistic regression
models are of the form
\begin{equation}\label{eq:1}
  P(y_i=1|{\mathbf{x}}_i)=p_i(\boldsymbol{\beta})
  =\frac{\exp({\mathbf{x}}_i^T\boldsymbol{\beta})}{1+\exp({\mathbf{x}}_i^T\boldsymbol{\beta})},\quad i=1,2, ..., n,
\end{equation}
where $y_i$'s $\in\{0,1\}$ are the responses and $\boldsymbol{\beta}$
is a $d\times1$ vector of unknown regression coefficients belonging to
a compact subset of $\mathbb{R}^d$. The unknown parameter
$\boldsymbol{\beta}$ is often estimated by the maximum likelihood
estimator (MLE) through maximizing the log-likelihood function with
respect to $\boldsymbol{\beta}$, namely,
\begin{equation}\label{log-l}
  \hat{{\boldsymbol{\beta}}}_{{{\textnormal{\tiny MLE}}}}=\arg\max_{{\boldsymbol{\beta}}}\ \ell({\boldsymbol{\beta}})
  =\arg\max_{{\boldsymbol{\beta}}}\ \sum_{i=1}^n\big[y_i\log p_i({\boldsymbol{\beta}})+(1-y_i)\log\{1- p_i({\boldsymbol{\beta}})\}\big].
\end{equation}

Analytically, there is no general closed-form solution to the MLE
$\hat{{\boldsymbol{\beta}}}_{{{\textnormal{\tiny MLE}}}}$, and
iterative procedures are often adopted to find it numerically. A
commonly used iterative procedure is Newton's method. Specifically for
logistic regression, Newton's method iteratively applies the following
formula until $\hat{{\boldsymbol{\beta}}}^{(t+1)}$ converges.
\begin{equation*}
  \hat{{\boldsymbol{\beta}}}^{(t+1)}=\hat{{\boldsymbol{\beta}}}^{(t)}
  +\left\{\sum_{i=1}^n w_i\Big(\hat{{\boldsymbol{\beta}}}^{(t)}\Big){\mathbf{x}}_i{\mathbf{x}}_i^T
  \right\}^{-1} \frac{\partial\ell\Big(\hat{{\boldsymbol{\beta}}}^{(t)}\Big)}{\partial{\boldsymbol{\beta}}},
\end{equation*}
where $w_i({\boldsymbol{\beta}})=p_i({\boldsymbol{\beta}})\{1-p_i({\boldsymbol{\beta}})\}$.  Since it requires
$O(nd^2)$ computing time in each iteration, the optimization procedure
takes $O(\zeta nd^2)$ time, where $\zeta$ is the number of iterations
required for the optimization procedure to converge.  One common
feature of the two motivating examples is their super-large sample
size. 
For such super-large sample problems, the computing time $O(nd^2)$ for
a single run may be too long to afford, let along to calculate it
iteratively. Therefore, computation is a bottleneck for the
application of logistic regression on massive data.

When proven statistical methods are no longer applicable due to
limited computing resources, a popular method to extract useful
information from data is the subsampling method
\citep{drineas2006sampling,mahoney2009cur,Drineas:11}. This approach
uses the estimate based on a subsample that is taken randomly from the
full data to approximate the estimate from the full data.  It is
termed {\it algorithmic leveraging} in \cite{PingMa2014-ICML,
  PingMa2014-JMLR} because the empirical statistical leverage scores
of the input covariate matrix are often used to define the nonuniform
subsampling probabilities.  There are numerous variants of subsampling
algorithms to solve the ordinary least squares (OLS) in linear
regression for large data sets, see \cite{drineas2006sampling,
  Drineas:11, PingMa2014-ICML,PingMa2014-JMLR, Ma2014}, among
others. Another strategy is to use random projections of data matrices
to fast approximate the OLS estimate, which was studied in
\cite{rokhlin2008fast}, \cite{dhillon2013new}, \cite{clarkson2013low}
and \cite{mcwilliams2014fast}. The aforementioned approaches have been
investigated exclusively within the context of linear regression, and
available results are mainly on algorithmic properties.  For logistic
regression, \cite{owen2007infinitely} derived interesting asymptotic
results for infinitely imbalanced data sets. \cite{king2001logistic}
investigated the problem of rare events data. \cite{fithian2014local}
proposed an efficient local case-control (LCC) subsampling method for
imbalanced data sets, in which the method was motivated by balancing
the subsample.
  In this paper, we focus on approximating the full data MLE using a
  subsample, and our method is motivated by minimizing the asymptotic
  mean squared error (MSE) of the resultant subsample-estimator given
  the full data.  We rigorously investigate the statistical
properties of the general subsampling estimator and obtain its
asymptotic distribution. More importantly, using this asymptotic
distribution, we derive \textbf{o}ptimal \textbf{s}ubsampling methods
\textbf{m}otivated from the \textbf{A}-optimality \textbf{c}riterion
(OSMAC) in the theory of optimal experimental design.

In this paper, we have two major contributions for theoretical and
methodological developments in subsampling for logistic regression
with massive data:
\begin{enumerate}
\item {\it Characterizations of optimal subsampling}.
  Most work on subsampling algorithms (under the context of linear
  regression) focuses on algorithmic issues. One exception is the work
  by \cite{PingMa2014-ICML,PingMa2014-JMLR}, in which expected values
  and variances of estimators from algorithmic leveraging were
  expressed approximately. However, 
   there was no precise theoretical
  investigation on when these approximations hold. In this paper, we
  rigorously prove that the resultant estimator from a general
  subsampling algorithm is consistent to the full data MLE, and
  establish the asymptotic normality of the resultant estimator.
  Furthermore, from the asymptotic distribution, we
  derive the optimal subsampling method that minimizes the asymptotic
  MSE or a weighted version of the asymptotic MSE.
\item {\it A novel two-step subsampling algorithm}.  The OSMAC that
  minimizes the asymptotic MSEs depends on the full data MLE
  $\hat{{\boldsymbol{\beta}}}_{{{\textnormal{\tiny MLE}}}}$, so the theoretical characterizations do not
  immediately translate into good algorithms. We propose a novel
  two-step algorithm to address this issue. The first step is to
  determine the importance score of each data point. In the
  second step, the importance scores are used to define nonuniform
  subsampling probabilities to be used for sampling from the full data
  set. { We prove that the estimator from the two-step algorithm is consistent and asymptotically normal with the optimal asymptotic covariance matrix under some optimality criterion.} The two-step subsampling algorithm  runs in $O(nd)$ time,
  whereas the full data MLE typically requires $O(\zeta nd^2)$ time to
  run. This improvement in computing time is much more significant
  than that obtained from applying the leverage-based subsampling
  algorithm to solve the OLS in linear regression. In linear
  regression, compared to a full data OLS which requires $O(nd^2)$
  time, the leverage-based algorithm with approximate leverage scores
  \citep{Drineas:12} requires $O(nd\log n/\varepsilon^2)$ time with
  $\varepsilon\in(0, 1/2]$, which is $o(nd^2)$ for the case of
  $\log n=o(d)$.
\end{enumerate}

The remainder of the paper is organized as follows. In
section~\ref{sec:random-subsampling}, we conduct a theoretical
analyses of a general subsampling algorithm for logistic
regression. In section~\ref{sec:appr-optim-subs}, we develop optimal
subsampling procedures to approximate the MLE in logistic
regression. A two-step algorithm is developed in
section~\ref{sec:two-step} to approximate these optimal subsampling
procedures, and its theoretical properties are studied. The empirical performance of our algorithms is evaluated
by numerical experiments on synthetic and real data sets in
Sections~\ref{sec:numerical-examples}. Section~\ref{sec:discussion}
summarizes the paper. Technical proofs for the theoretical results, as
well as additional numerical experiments are given in the
Supplementary Materials. 

\section{General Subsampling Algorithm and its Asymptotic Properties}
\label{sec:random-subsampling} In this section, we first present a
general subsampling algorithm for approximating $\hat{{\boldsymbol{\beta}}}_{{{\textnormal{\tiny MLE}}}}$,
and then establish the consistency and asymptotic normality of the
resultant estimator. Algorithm~\ref{alg:1} describes the general
subsampling procedure.
\begin{algorithm}
  \begin{itemize}
  \item \textbf{Sampling:} Assign subsampling probabilities $\pi_i$,
    $i=1,2,...n$, for all data points.  Draw a random subsample of
    size $r\ (\ll n)$, according to the probabilities
    $\{\pi_i\}_{i=1}^{n}$, from the full data. Denote the covariates,
    responses, and subsampling probabilities in the subsample as
    ${\mathbf{x}}^{*}_i$, $y^{*}_i$, and $\pi_i^{*}$, respectively, for
    $i=1,2,...,r$.
  \item \textbf{Estimation:} Maximize the following weighted
    log-likelihood function to get the estimate $\tilde{{\boldsymbol{\beta}}}$ based
    on the subsample.
    \begin{equation*}
      \ell^*({\boldsymbol{\beta}})=\frac{1}{r}\sum_{i=1}^r\frac{1}{\pi_i^*}[y_i^*\log p_i^*({\boldsymbol{\beta}})+(1-y_i^*)\log\{1-p_i^*({\boldsymbol{\beta}})\}],
    \end{equation*}
    where
    $p_i^*({\boldsymbol{\beta}})={\exp({\boldsymbol{\beta}}^T{\mathbf{x}}_i^*)}/\{1+\exp
    ({\boldsymbol{\beta}}^T{\mathbf{x}}_i^*)\}$.
    Due to the convexity of $\ell^*({\boldsymbol{\beta}})$, the maximization can be
    implemented by Newton's method, i.e., iteratively applying the
    following formula until $\tilde{{\boldsymbol{\beta}}}^{(t+1)}$ and
    $\tilde{{\boldsymbol{\beta}}}^{(t)}$ are close enough,
    \begin{equation}\label{eq:21}
      \tilde{{\boldsymbol{\beta}}}^{(t+1)}=\tilde{{\boldsymbol{\beta}}}^{(t)}
      +\left\{\sum_{i=1}^r\frac{w_i^*\big(\tilde{{\boldsymbol{\beta}}}^{(t)}\big)
          {\mathbf{x}}_i^*({\mathbf{x}}_i^*)^T}{\pi_i^*}\right\}^{-1}
      \sum_{i=1}^r\frac{\big\{y_i^*-p_i^*\big(\tilde{{\boldsymbol{\beta}}}^{(t)}\big)\big\}
        {\mathbf{x}}_i^*}{\pi_i^*},
    \end{equation}
    where $w_i^*({\boldsymbol{\beta}})=p_i^*({\boldsymbol{\beta}})\{1-p_i^*({\boldsymbol{\beta}})\}$.
  \end{itemize}
  \caption{General subsampling algorithm}
  \label{alg:1}
\end{algorithm}

Now, we investigate asymptotic properties of this general subsampling
algorithm, which provide guidance on how to develop algorithms with
better approximation qualities. 
Note that in the two motivating examples, the sample sizes are super-large, but
the numbers of predictors are unlikely to increase even if the sample sizes further increase. 
We assume that $d$ is fixed and
  $n\rightarrow\infty$. For easy of discussion, we assume that ${\mathbf{x}}_i$'s are independent and identically
distributed (i.i.d) with the same distribution as that of ${\mathbf{x}}$. The case of
nonrandom ${\mathbf{x}}_i$'s is presented in the Supplementary Materials. To facilitate
the presentation, denote the full data matrix as ${\mathcal{F}_n}=({\mathbf{X}}, {\mathbf{y}})$, where
${\mathbf{X}}=({\mathbf{x}}_1, {\mathbf{x}}_2,\dots,{\mathbf{x}}_n)^T$ is the covariate matrix and
${\mathbf{y}}=(y_1,y_2,\dots,y_n)^T$ is the vector of responses. 
Throughout the paper, $\|\mathbf{v}\|$ denotes the Euclidean norm of a
vector $\mathbf{v}$, i.e.,
$\|\mathbf{v}\|=(\mathbf{v}^T\mathbf{v})^{1/2}$. We need the following
assumptions to establish the first asymptotic result.
\begin{assumption}\label{asp:1}
  As $n\rightarrow\infty$,
  ${\mathbf{M}}_X=n^{-1}\sum_{i=1}^n w_i(\hat{{\boldsymbol{\beta}}}_{{{\textnormal{\tiny MLE}}}}){\mathbf{x}}_i{\mathbf{x}}_i^T$
  goes to a positive-definite matrix in probability and
  $n^{-1}\sum_{i=1}^n\|{\mathbf{x}}_i\|^3=O_P(1)$.
\end{assumption}
\begin{assumption}\label{asp:2}
  $n^{-2}\sum_{i=1}^n\pi_i^{-1}\|{\mathbf{x}}_i\|^k=O_P(1)$ for $k=2, 4$.
\end{assumption}
Assumption~\ref{asp:1} imposes two conditions on the covariate
distribution and this assumption holds if ${\mathrm{E}}({\mathbf{x}}{\mathbf{x}}^T)$ is positive
definite and ${\mathrm{E}}\|{\mathbf{x}}\|^3<\infty$. Assumption~\ref{asp:2} is a
condition on both subsampling probabilities and the covariate
distribution. For uniform subsampling with $\pi_i=n^{-1}$, a
sufficient condition for this assumption is that ${\mathrm{E}}\|{\mathbf{x}}\|^4<\infty$.

The theorem below presents the consistency of the estimator from the
subsampling algorithm to the full data MLE.
\begin{theorem}\label{thm:1}
  If assumptions \ref{asp:1} and \ref{asp:2} hold, then as
  $n\rightarrow\infty$ and $r\rightarrow\infty$, $\tilde{{\boldsymbol{\beta}}}$ is
  consistent to $\hat{{\boldsymbol{\beta}}}_{{{\textnormal{\tiny MLE}}}}$ in conditional probability,
  given ${\mathcal{F}_n}$ in probability. Moreover, the rate of convergence is
  $r^{-1/2}$. That is, with probability approaching one, for any
  $\epsilon>0$, there exists a finite $\Delta_\epsilon$ and
  $r_\epsilon$ such that
  \begin{equation}\label{eq:18}
    P(\|\tilde{{\boldsymbol{\beta}}}-\hat{{\boldsymbol{\beta}}}_{{{\textnormal{\tiny MLE}}}}\|\ge
    r^{-1/2}\Delta_\epsilon|{\mathcal{F}_n})<\epsilon
  \end{equation}
  for all $r>r_\epsilon$.
\end{theorem}
The consistency result shows that the approximation error can be made
as small as possible by a large enough subsample size $r$, as the
approximation error is at the order of $O_{P|{\mathcal{F}_n}}(r^{-1/2})$. Here
the probability measure in $O_{P|{\mathcal{F}_n}}(\cdot)$ is the conditional
measure given ${\mathcal{F}_n}$. This result has some similarity to the
finite-sample result of the worst-case error bound for arithmetic
leveraging in linear regression \citep{Drineas:11}, but neither of
them gives the full distribution of the approximation error.

Besides consistency, we derive the asymptotic distribution of the
approximation error, and prove that the approximation error,
$\tilde{{\boldsymbol{\beta}}}-\hat{{\boldsymbol{\beta}}}_{{{\textnormal{\tiny MLE}}}}$, is asymptotically normal. To obtain
this result, we need an additional assumption below, which is required
by the Lindeberg-Feller central limit theorem.
\begin{assumption}\label{asp:3}
  There exists some $\delta>0$ such that
  $n^{-(2+\delta)}\sum_{i=1}^n{\pi_i^{-1-\delta}}{\|{\mathbf{x}}_i\|^{2+\delta}}=O_P(1)$.
\end{assumption}
The aforementioned three assumptions are essentially moment conditions
and are very general. For example, a sub-Gaussian distribution
\citep{buldygin1980sub} has finite moment generating function on
$\mathbb{R}$ and thus has finite moments up to any finite order. If
the distribution of each component of ${\mathbf{x}}$ belongs to the class of
sub-Gaussian distributions and the covariance matrix of ${\mathbf{x}}$ is
positive-definite, then all the conditions are satisfied by the
subsampling probabilities considered in this paper. The result of
asymptotic normality is presented in the following theorem.
\begin{theorem}\label{thm:2}
  If assumptions \ref{asp:1}, \ref{asp:2}, and \ref{asp:3} hold, then
  as $n\rightarrow\infty$ and $r\rightarrow\infty$, conditional on
  ${\mathcal{F}_n}$ in probability,
  \begin{equation}\label{normal}
    {\mathbf{V}}^{-1/2}(\tilde{{\boldsymbol{\beta}}}-\hat{{\boldsymbol{\beta}}}_{{{\textnormal{\tiny MLE}}}})
    \longrightarrow N(0,\mathbf{I})
  \end{equation}
  in distribution, where
  \begin{equation}\label{varorder}
    {\mathbf{V}}={\mathbf{M}}_X^{-1}{\mathbf{V}}_c{\mathbf{M}}_X^{-1}=O_{p}(r^{-1})
  \end{equation}
  and
  \begin{equation}
    {\mathbf{V}}_c=\frac{1}{rn^2}
    \sum_{i=1}^n\frac{\{y_i-p_i(\hat{{\boldsymbol{\beta}}}_{{{\textnormal{\tiny MLE}}}})\}^2{\mathbf{x}}_i{\mathbf{x}}_i^T}{\pi_i}.
    \label{eq:vc}
  \end{equation}
\end{theorem}
\begin{remark}
  Note that in Theorems~\ref{thm:1} and \ref{thm:2} we are
  approximating the full data MLE, and the results hold for the case
  of oversampling ($r>n$). However, this scenario is not practical
  because it is more computationally intense than using the full
  data. Additionally, the distance between $\hat{{\boldsymbol{\beta}}}_{{{\textnormal{\tiny MLE}}}}$ and
  ${\boldsymbol{\beta}}_0$, the true parameter, is at the order of $O_P(n^{-1/2})$.
  Oversampling does not result in any gain in terms of estimating the
  true parameter. For aforementioned reasons, the scenario of
  oversampling is not of our interest and we focus on the scenario
  that $r$ is much smaller than $n$, typically, $n-r\rightarrow\infty$
  or $r/n\rightarrow0$.
\end{remark}

Result~\eqref{normal} shows that the distribution of
$\tilde{{\boldsymbol{\beta}}}-\hat{{\boldsymbol{\beta}}}_{{{\textnormal{\tiny MLE}}}}$ given ${\mathcal{F}_n}$ can be approximated
by that of ${\mathbf{u}}$, a normal random variable with distribution
$N({\bf 0},{\mathbf{V}})$. In other words, the probability
$P(r^{1/2}\|\tilde{{\boldsymbol{\beta}}}-\hat{{\boldsymbol{\beta}}}_{{{\textnormal{\tiny MLE}}}}\|\ge \Delta|{\mathcal{F}_n})$ can
be approximated by $P(r^{1/2}\|{\mathbf{u}}\|\ge\Delta|{\mathcal{F}_n})$ for any
$\Delta$. To facilitate the discussion, we write result~\eqref{normal} as
\begin{equation}\label{asnormal}
  \tilde{{\boldsymbol{\beta}}}-\hat{{\boldsymbol{\beta}}}_{{{\textnormal{\tiny MLE}}}}|{\mathcal{F}_n}\overset{a}{\sim} {\mathbf{u}},
\end{equation}
where $\overset{a}{\sim}$ means the distributions of the two terms are
asymptotically the same. This result is more statistically informative
than a worst-case error bound for the approximation error
$\tilde{{\boldsymbol{\beta}}}-\hat{{\boldsymbol{\beta}}}_{{{\textnormal{\tiny MLE}}}}$. Moreover, this result gives
direct guidance on how to reduce the approximation error while an
error bound does not, because a smaller bound does not necessarily mean
a smaller approximation error.

Although the distribution of $\tilde{{\boldsymbol{\beta}}}-\hat{{\boldsymbol{\beta}}}_{{{\textnormal{\tiny MLE}}}}$
given ${\mathcal{F}_n}$ can be approximated by that of ${\mathbf{u}}$, this does not
necessarily imply that
${\mathrm{E}}(\|\tilde{{\boldsymbol{\beta}}}-\hat{{\boldsymbol{\beta}}}_{{{\textnormal{\tiny MLE}}}}\|^2|{\mathcal{F}_n})$ is close to
${\mathrm{E}}(\|{\mathbf{u}}\|^2|{\mathcal{F}_n})$. ${\mathrm{E}}(\|{\mathbf{u}}\|^2|{\mathcal{F}_n})$ is an asymptotic MSE (AMSE)
of $\tilde{\boldsymbol{\beta}}$ and it is always well defined. However, rigorously
speaking, ${\mathrm{E}}(\|\tilde{{\boldsymbol{\beta}}}-\hat{{\boldsymbol{\beta}}}_{{{\textnormal{\tiny MLE}}}}\|^2|{\mathcal{F}_n})$, or any
conditional moment of $\tilde{{\boldsymbol{\beta}}}$, is undefined, because there is
a nonzero probability that $\tilde{{\boldsymbol{\beta}}}$ based on a subsample does
not exist. The same problem exists in subsampling estimators for the
OLS in linear regression. To address this issue, we define
$\tilde{{\boldsymbol{\beta}}}$ to be {\bf 0} when the MLE based on a subsample does
not exist. Under this definition, if $\tilde{{\boldsymbol{\beta}}}$ is uniformly
integrable under the conditional measure given ${\mathcal{F}_n}$,
$r^{1/2} \{{\mathrm{E}}(\|\tilde{{\boldsymbol{\beta}}} - \hat{{\boldsymbol{\beta}}}_{{{\textnormal{\tiny MLE}}}}\|^2|{\mathcal{F}_n}) -
{\mathrm{E}}(\|{\mathbf{u}}\|^2|{\mathcal{F}_n})\} \rightarrow0$ in probability.

Results in Theorems~\ref{thm:1} and \ref{thm:2} are distributional
results conditional on the observed data, which fulfill our primary
goal of approximating the full data MLE $\hat{{\boldsymbol{\beta}}}_{{{\textnormal{\tiny MLE}}}}$.
Conditional inference is quite common in statistics, and the most
popular method is the Bootstrap \citep{Efron:79,
  efron1994introduction}. 
The Bootstrap (nonparametric) is the uniform subsampling approach with
subsample size equaling the full data sample size. If $\pi_i=1/n$ and
$r=n$, then results in Theorems~\ref{thm:1} and \ref{thm:2} reduce to
the asymptotic results for the Bootstrap.  However, the Bootstrap and
the subsampling method in the paper have very distinct goals. The
Bootstrap focuses on approximating complicated distributions and are
used when explicit solutions are unavailable, while the subsampling
method considered here has a primary motivation to achieve feasible
computation and is used even closed-form solutions are available.

\section{Optimal Subsampling Strategies}
\label{sec:appr-optim-subs}
To implement Algorithm 1, one has to specify the subsampling
probability (SSP) $\boldsymbol{\pi}=\{\pi_i\}_{i=1}^{n}$ for the full
data. An easy choice is to use the uniform SSP
$\boldsymbol{\pi}^{\mathrm{UNI}}=\{\pi_i=n^{-1}\}_{i=1}^{n}$. However, an
algorithm with the uniform SSP may not be ``optimal'' and a nonuniform
SSP may have a better performance. In this section, we propose more
efficient subsampling procedures by choosing nonuniform $\pi_i$'s to
``minimize'' the asymptotic variance-covariance matrix ${\mathbf{V}}$ in
\eqref{varorder}. However, since ${\mathbf{V}}$ is a matrix, the meaning of
``minimize'' needs to be defined. We adopt the idea of the
$A$-optimality from optimal design of experiments and use the trace of
a matrix to induce a complete ordering of the variance-covariance
matrices \citep{kiefer1959}. It turns out that this approach is
equivalent to minimizing the asymptotic MSE of the resultant
estimator. Since this \textbf{o}ptimal \textbf{s}ubsampling procedure
is \textbf{m}otivated from the \textbf{A}-optimality
\textbf{c}riterion, we call our method the OSMAC.

\subsection{Minimum Asymptotic MSE of $\tilde{\boldsymbol{\beta}}$}\label{sec:mMSE}
From the result in Theorem~\ref{thm:2}, the asymptotic MSE of
$\tilde{{\boldsymbol{\beta}}}$ is equal to the trace of ${\mathbf{V}}$, namely,
\begin{equation} \label{eq:22}
  \mathrm{AMSE}(\tilde{{\boldsymbol{\beta}}})={\mathrm{E}}(\|{\mathbf{u}}\|^2|{\mathcal{F}_n})=\mathrm{tr}({\mathbf{V}}).
\end{equation}
From \eqref{varorder}, ${\mathbf{V}}$ depends on $\{\pi_i\}_{i=1}^n$, and clearly,
$\{\pi_i=n^{-1}\}_{i=1}^{n}$ may not produce the smallest value of
$\mathrm{tr}({\mathbf{V}})$. The key idea of optimal subsampling is to choose nonuniform
SSP such that the $\mathrm{AMSE}(\tilde{{\boldsymbol{\beta}}})$ in \eqref{eq:22} is
minimized.  Since minimizing the trace of the (asymptotic)
variance-covariance matrix is called the $A$-optimality criterion
\citep{kiefer1959}, the resultant SSP is $A$-optimal in the language
of optimal design. The following theorem gives the $A$-optimal SSP
that minimizes the asymptotic MSE of $\tilde{\boldsymbol{\beta}}$.
\begin{theorem}\label{thm:3}
  In Algorithm 1, if the SSP is chosen such that
  \begin{equation}\label{eq:19}
    \pi_i^{\mathrm{mMSE}}=
    \frac{|y_i-p_i(\hat{{\boldsymbol{\beta}}}_{{{\textnormal{\tiny MLE}}}})|\|{\mathbf{M}}_X^{-1}{\mathbf{x}}_i\|}
    {\sum_{j=1}^n|y_j-p_j(\hat{{\boldsymbol{\beta}}}_{{{\textnormal{\tiny MLE}}}})|\|{\mathbf{M}}_X^{-1}{\mathbf{x}}_j\|},\;
    i=1,2,...,n,
  \end{equation}
  then the asymptotic MSE of $\tilde{\boldsymbol{\beta}}$, $\mathrm{tr}({\mathbf{V}})$, attains its
  minimum.
\end{theorem}
As observed in (\ref{eq:19}), the optimal SSP
$\boldsymbol{\pi}^{\mathrm{mMSE}}=\{\pi_{i}^{\mathrm{mMSE}}\}_{i=1}^{n}$
depends on data through both the covariates and the responses
directly.  For the covariates, the optimal SSP is larger for a larger
$\|{\mathbf{M}}_X^{-1}{\mathbf{x}}_i\|$, which is the square root of the $i$th diagonal
element of the matrix ${\mathbf{X}}{\mathbf{M}}_X^{-2}{\mathbf{X}}^T$.
{ The effect of the responses on the optimal SSP depends on
discrimination difficulties through the term
$|y_i-p_i(\hat{{\boldsymbol{\beta}}}_{{{\textnormal{\tiny MLE}}}})|$. Interestingly, if the full data MLE
$\hat{{\boldsymbol{\beta}}}_{{{\textnormal{\tiny MLE}}}}$ in $|y_i-p_i(\hat{{\boldsymbol{\beta}}}_{{{\textnormal{\tiny MLE}}}})|$ is replace
by a pilot estimate, then this term is exactly the same as the
probability in the local case-control (LCC) subsampling procedure in dealing with imbalanced data \citep{fithian2014local}. However, Poisson sampling and unweighted MLE were used in the LCC subsampling procedure.

To see the effect of the responses on the optimal SSP, let}
 $S_0=\{i:\ y_i=0\}$ and
$S_1=\{i:\ y_i=1\}$. The effect of $p_i(\hat{{\boldsymbol{\beta}}}_{{{\textnormal{\tiny MLE}}}})$ on
$\pi_i^{\mathrm{mMSE}}$ is positive for the $S_0$ set, i.e. a larger
$p_i(\hat{{\boldsymbol{\beta}}}_{{{\textnormal{\tiny MLE}}}})$ results in a larger
$\pi_i^{\mathrm{mMSE}}$, while the effect is negative for
the $S_1$ set, i.e. a larger $p_i(\hat{{\boldsymbol{\beta}}}_{{{\textnormal{\tiny MLE}}}})$ results in a
smaller $\pi_i^{\mathrm{mMSE}}$. The optimal subsampling approach is
more likely to select data points with smaller
$p_i(\hat{{\boldsymbol{\beta}}}_{{{\textnormal{\tiny MLE}}}})$'s when $y_i$'s are 1 and data points with
larger $p_i(\hat{{\boldsymbol{\beta}}}_{{{\textnormal{\tiny MLE}}}})$'s when $y_i$'s are 0. Intuitively,
it attempts to give preferences to data points that are more likely to
be mis-classified. This can also be seen in the expression of
$\mathrm{tr}({\mathbf{V}})$. From~\eqref{varorder} and \eqref{eq:vc},
\begin{align*}
  \mathrm{tr}({\mathbf{V}})=&\mathrm{tr}({\mathbf{M}}_X^{-1}{\mathbf{V}}_c{\mathbf{M}}_X^{-1})\\
  =&\frac{1}{rn^2}\mathrm{tr}\left[
     \sum_{i=1}^n\frac{\{y_i-p_i(\hat{{\boldsymbol{\beta}}}_{{{\textnormal{\tiny MLE}}}})\}^2
     {\mathbf{M}}_X^{-1}{\mathbf{x}}_i{\mathbf{x}}_i^T{\mathbf{M}}_X^{-1}}{\pi_i}\right]\\
  =&\frac{1}{rn^2}\sum_{i=1}^n\frac{\{y_i-p_i(\hat{{\boldsymbol{\beta}}}_{{{\textnormal{\tiny MLE}}}})\}^2
     \mathrm{tr}({\mathbf{M}}_X^{-1}{\mathbf{x}}_i{\mathbf{x}}_i^T{\mathbf{M}}_X^{-1})}{\pi_i}\\
  =&\frac{1}{rn^2}\sum_{i\in S_0}
     \frac{\{p_i(\hat{{\boldsymbol{\beta}}}_{{{\textnormal{\tiny MLE}}}})\}^2\|{\mathbf{M}}_X^{-1}{\mathbf{x}}_i\|^2}{\pi_i}
     +\frac{1}{rn^2}\sum_{i\in S_1}
     \frac{\{1-p_i(\hat{{\boldsymbol{\beta}}}_{{{\textnormal{\tiny MLE}}}})\}^2\|{\mathbf{M}}_X^{-1}{\mathbf{x}}_i\|^2}{\pi_i}.
\end{align*}

From the above equation, a larger value of $p_i(\hat{{\boldsymbol{\beta}}}_{{{\textnormal{\tiny MLE}}}})$
results in a larger value of the summation for the $S_0$ set, so a
larger value is assigned to $\pi_i$ to reduce this summation. On the
other hand for the $S_1$ set, a larger value of
$p_i(\hat{{\boldsymbol{\beta}}}_{\textnormal{\tiny{ MLE}}})$ results in a smaller
value of the summation, so a smaller value is assigned to
$\pi_i$.

The optimal subsampling approach also echos the result in
\cite{Silvapulle1981}, which gave a necessary and sufficient condition
for the existence of the MLE in logistic regression. To see this, let
\begin{equation*}
  F_0=\left\{\sum_{i\in S_0}k_i{\mathbf{x}}_i\Big|k_i>0\right\}
  \quad \text{ and }\quad
  F_1=\left\{\sum_{i\in S_1}k_i{\mathbf{x}}_i\Big|k_i>0\right\}.
\end{equation*}
Here, $F_0$ and $F_1$ are convex cones generated by covariates in the
$S_0$ and the $S_1$ sets, respectively.  \cite{Silvapulle1981} showed
that the MLE in logistic regression is uniquely defined if and only if
$F_0\cap F_1\neq\phi$, where $\phi$ is the empty set. From Theorem II
in \cite{Dines1926}, $F_0\cap F_1\neq\phi$ if and only if there does
{\it not} exist a ${\boldsymbol{\beta}}$ such that
\begin{equation}\label{eq:23}
  {\mathbf{x}}_i^T{\boldsymbol{\beta}}\le0 \text{ for all } i\in S_0, \quad
  {\mathbf{x}}_i^T{\boldsymbol{\beta}}\ge0 \text{ for all } i\in S_1,
\end{equation}
and at least one strict inequality holds.
The statement in \eqref{eq:23} is equivalent to the following statement in \eqref{eq:25} below.
\begin{equation}\label{eq:25}
  p_i({\boldsymbol{\beta}})\le0.5 \text{ for all } i\in S_0, \quad
  p_i({\boldsymbol{\beta}})\ge0.5 \text{ for all } i\in S_1.
\end{equation}

This means if there exist a ${\boldsymbol{\beta}}$ such that
$\{p_i({\boldsymbol{\beta}}),i\in S_0\}$ and $\{p_i({\boldsymbol{\beta}}),i\in S_1\}$ can be
separated, then the MLE does not exist. The optimal subsampling SSP
strives to increase the overlap of these two sets in the direction of
$p_i(\hat{{\boldsymbol{\beta}}}_{{{\textnormal{\tiny MLE}}}})$. Thus it decreases the probability of the
scenario that the MLE does not exist based on a resultant subsample.

\subsection{Minimum Asymptotic MSE of
  ${\mathbf{M}}_X\tilde{\boldsymbol{\beta}}$}\label{sec:mVc}
\label{sec:optimal-choice-b}
The optimal SSPs derived in the previous section require the
calculation of $\|{\mathbf{M}}_X^{-1}{\mathbf{x}}_i\|$ for $i=1,2,...,n$, which takes
$O(nd^2)$ time. In this section, we propose a modified optimality
criterion, under which calculating the optimal SSPs requires less
time.

To motivate the optimality criteria, we need to define the partial
ordering of positive definite matrices. For two positive definite
matrices $\mathbf{A}_1$ and $\mathbf{A}_2$,
$\mathbf{A}_1\ge \mathbf{A}_2$ if and only if
$\mathbf{A}_1-\mathbf{A}_2$ is a nonnegative definite matrix. This
definition is called the Loewner-ordering. Note that ${\mathbf{V}}={\mathbf{M}}_X^{-1}{\mathbf{V}}_c{\mathbf{M}}_X^{-1}$ in
\eqref{varorder} depends on $\boldsymbol{\pi}$ through ${\mathbf{V}}_c$ in
\eqref{eq:vc}, and ${\mathbf{M}}_X$ does not depend on $\boldsymbol{\pi}$. For two given
SSPs $\boldsymbol{\pi}^{(1)}$ and $\boldsymbol{\pi}^{(2)}$,
${\mathbf{V}}(\boldsymbol{\pi}^{(1)})\leq {\mathbf{V}}(\boldsymbol{\pi}^{(2)})$ if and only if
${\mathbf{V}}_c(\boldsymbol{\pi}^{(1)})\leq {\mathbf{V}}_c(\boldsymbol{\pi}^{(2)})$.
This gives us guidance to simplify the optimality criterion. Instead of focusing on the more complicated matrix ${\mathbf{V}}$, we define an alternative optimality criterion by focusing on ${\mathbf{V}}_c$. Specifically, instead of minimizing $\mathrm{tr}({\mathbf{V}})$ as in Section~\ref{sec:mMSE}, we choose to minimize $\mathrm{tr}({\mathbf{V}}_c)$. The primary goal of this alternative  optimality criterion is to further reduce the computing time. 

The following theorem gives the optimal SSP that minimizes the trace of ${\mathbf{V}}_c$.
\begin{theorem}\label{thm:4}
  In Algorithm 1, if the SSP is chosen such that
  \begin{equation}\label{optimalBpi}
    \pi_i^{\mathrm{mVc}}=\frac{|y_i-p_i(\hat{{\boldsymbol{\beta}}}_{{{\textnormal{\tiny MLE}}}})|\|{\mathbf{x}}_i\|}
    {\sum_{j=1}^n|y_j-p_j(\hat{{\boldsymbol{\beta}}}_{{{\textnormal{\tiny MLE}}}})|\|{\mathbf{x}}_j\|},\ i=1,2, ..., n,
  \end{equation}
  then $\mathrm{tr}({\mathbf{V}}_c)$, attains its minimum.
\end{theorem}

It turns out that the alternative optimality criterion indeed greatly reduces the computing time. From Theorem~\ref{thm:4}, the effect of the covariates on
$\boldsymbol{\pi}^{\mathrm{mVc}}=\{\pi_i^{\mathrm{mVc}}\}_{i=1}^{n}$ is presented by $\|{\mathbf{x}}_i\|$, instead of $\|{\mathbf{M}}_X^{-1}{\mathbf{x}}_i\|$ as in $\boldsymbol{\pi}^{\mathrm{mMSE}}$. The computational benefit is obvious: it requires $O(nd)$ time to calculate $\|{\mathbf{x}}_i\|$ for $i=1,2,...,n$, which is significantly less than the required $O(nd^2)$ time to calculate
$\|{\mathbf{M}}_X^{-1}{\mathbf{x}}_i\|$ for $i=1,2,...,n$. 

Besides the computational benefit, this alternative criterion also enjoys  nice interpretations from the following aspects. First, the term $|y_i-p_i(\hat{{\boldsymbol{\beta}}}_{{{\textnormal{\tiny MLE}}}})|$ functions the same as in the case of $\boldsymbol{\pi}^{\mathrm{mMSE}}$. Hence all the nice interpretations and properties related to this term for $\boldsymbol{\pi}^{\mathrm{mMSE}}$ in Section~\ref{sec:mMSE} are true for $\boldsymbol{\pi}^{\mathrm{mVc}}$ in Theorem~\ref{thm:4}. Second, from \eqref{asnormal},
\begin{equation*}
  {\mathbf{M}}_X(\tilde{{\boldsymbol{\beta}}}-\hat{{\boldsymbol{\beta}}}_{{{\textnormal{\tiny MLE}}}})\big|{\mathcal{F}_n}\overset{a}
  {\sim} {\mathbf{M}}_X{\mathbf{u}}, \text{ where } {\mathbf{M}}_X{\mathbf{u}}\sim N({\bf 0},{\mathbf{V}}_c) \text{ given }{\mathcal{F}_n}.
\end{equation*}
This shows that $\mathrm{tr}({\mathbf{V}}_c)={\mathrm{E}}(\|{\mathbf{M}}_X{\mathbf{u}}\|^2|{\mathcal{F}_n})$ is the AMSE of
${\mathbf{M}}_X\tilde{{\boldsymbol{\beta}}}$ in approximating
${\mathbf{M}}_X\hat{{\boldsymbol{\beta}}}_{{{\textnormal{\tiny MLE}}}}$. Therefore, the SSP
$\boldsymbol{\pi}^{\mathrm{mVc}}$ is optimal in terms of minimizing the AMSE
of ${\mathbf{M}}_X\tilde{{\boldsymbol{\beta}}}$. Third, the alternative criterion  also
corresponds to the commonly used linear optimality (L-optimality)
criterion in optimal experimental design \citep[c.f. Chapter 10
of][]{atkinson2007optimum}. The L-optimality criterion minimizes the
trace of the product of the asymptotic variance-covariance matrix and
a constant matrix. Its aim is to improve the quality of prediction in
linear regression. For our problem, note that
$\mathrm{tr}({\mathbf{V}}_c)=\mathrm{tr}({\mathbf{M}}_X{\mathbf{V}}{\mathbf{M}}_X)=\mathrm{tr}({\mathbf{V}}{\mathbf{M}}_X^2)$ and ${\mathbf{V}}$ is the asymptotic
variance-covariance matrix of $\tilde{{\boldsymbol{\beta}}}$, so the SSP $\boldsymbol{\pi}^{\mathrm{mVc}}$ is L-optimal in the language of optimal design.

\section{Two-Step Algorithm}
\label{sec:two-step}
The SSPs in \eqref{eq:19} and \eqref{optimalBpi} depend on
$\hat{{\boldsymbol{\beta}}}_{{{\textnormal{\tiny MLE}}}}$, which is the full data MLE to be approximated,
so an exact OSMAC is not applicable directly.  We propose a two-step
algorithm to approximate the OSMAC. { In the first step, a
  subsample of $r_0$ is taken to get a pilot estimate of
  $\hat{{\boldsymbol{\beta}}}_{{{\textnormal{\tiny MLE}}}}$, which is then used to approximate the optimal
  SSPs for drawing the more informative second step subsample.} The
two-step algorithm is presented in Algorithm~\ref{alg:2}.

\begin{algorithm}
  \caption{Two-step Algorithm}\label{alg:2}
  \begin{itemize}
  \item \textbf{Step 1:} Run Algorithm~\ref{alg:1} with subsample size $r_0$ to
    obtain an estimate $\tilde{{\boldsymbol{\beta}}}_0$, using either the uniform
    SSP $\boldsymbol{\pi}^{\mathrm{UNI}}=\{n^{-1}\}_{i=1}^{n}$ or SSP
    $\{\pi_i^{\mathrm{prop}}\}_{i=1}^{n}$, where
    $\pi_i^{\mathrm{prop}}=(2n_0)^{-1}$ if $i\in S_0$ and
    $\pi_i^{\mathrm{prop}}=(2n_1)^{-1}$ if $i\in S_1$. Here, $n_0$ and
    $n_1$ are the numbers of elements in sets $S_0$ and $S_1$,
    respectively. Replace $\hat{{\boldsymbol{\beta}}}_{{{\textnormal{\tiny MLE}}}}$ with
    $\tilde{{\boldsymbol{\beta}}}_0$ in \eqref{eq:19} or \eqref{optimalBpi} to get
    an approximate optimal SSP corresponding to a chosen optimality
    criterion.
  \item \textbf{Step 2:} Subsample with replacement for a subsample of
    size ${r}$ with the approximate optimal SSP calculated in Step 1.
    Combine the samples from the two steps and obtain the estimate
    ${\breve{{\boldsymbol{\beta}}}}$ based on the total subsample of size $r_0+r$
    according to the Estimation step in Algorithm~\ref{alg:1}.
  \end{itemize}
\end{algorithm}

\begin{remark}
  In Step 1, for the $S_0$ and $S_1$ sets, different subsampling
  probabilities can be specified, each of which is equal to half of
  the inverse of the set size. The purpose is to balance the numbers
  of 0's and 1's in the responses for the subsample. If the full data
  is very imbalanced, the probability that the MLE exists for a
  subsample obtained using this approach is higher than that for a
  subsample obtained using uniform subsampling. This procedure is
  called the case-control sampling
  \citep{scott1986fitting,fithian2014local}. If the proportion of 1's
  is close to 0.5, the uniform SSP is preferable in Step 1 due to its
  simplicity.
\end{remark}
\begin{remark}
  As shown in Theorem \ref{thm:1}, $\tilde{{\boldsymbol{\beta}}}_0$ from Step 1
  approximates $\hat{{\boldsymbol{\beta}}}_{{{\textnormal{\tiny MLE}}}}$ accurately as long as $r_0$ is
  not too small. On the other hand, the efficiency of the two-step
  algorithm would decrease, if $r_0$ gets close to the total subsample size $r_0+r$ and ${r}$ is
  relatively small. We will need $r_0$ to be a small term compared with $r^{1/2}$, i.e., $r_0=o(r^{-1/2})$, in order to prove the consistency and asymptotically optimality of the two-step algorithm in Section~\ref{sec:asympt-prop}. 
\end{remark}

Algorithm~2 greatly reduces the computational cost compared to using
the full data. The major computing time is to approximate the optimal
SSPs which does not require iterative calculations on the full data.
Once the approximately optimal SSPs are available, the time to obtain
$\breve{{\boldsymbol{\beta}}}$ in the second step is $O({\zeta}rd^2)$ where
${\zeta}$ is the number of iterations of the iterative procedure in
the second step.  If the $S_0$ and $S_1$ sets are not separated, the
time to obtain $\tilde{{\boldsymbol{\beta}}}_0$ in the first step is
$O(n+\zeta_0r_0d^2)$ where $\zeta_0$ is the number of iterations of
the iterative procedure in the first step. To calculate the estimated
optimal SSPs, the required times are different for different optimal
SSPs. For ${\pi}_i^{\mathrm{mVc}}$, $i=1,...n$,
the required time is $O(nd)$. For ${\pi}_i^{\mathrm{mMSE}}$, $i=1,...,n$,
the required time is longer because they involve
${\mathbf{M}}_X=n^{-1}\sum_{i=1}^n w_i(\hat{{\boldsymbol{\beta}}}_{{{\textnormal{\tiny MLE}}}}){\mathbf{x}}_i{\mathbf{x}}_i^T$. If
$\hat{{\boldsymbol{\beta}}}_{{{\textnormal{\tiny MLE}}}}$ is replaced by $\tilde{{\boldsymbol{\beta}}}_0$ in
$w_i(\hat{{\boldsymbol{\beta}}}_{{{\textnormal{\tiny MLE}}}})$ and then the full data is used to calculate
an estimate of ${\mathbf{M}}_X$, the required time is $O(nd^2)$. Note that
${\mathbf{M}}_X$ can be estimated by
$\tilde{{\mathbf{M}}}_X^0 =(nr_0)^{-1}\sum_{i=1}^{r_0}(\pi_i^*)^{-1} w_i^*(\tilde{\boldsymbol{\beta}}_0)
{\mathbf{x}}_i^*({\mathbf{x}}_i^*)^T$
based on the selected subsample, for which the calculation only
requires $O(r_0d^2)$ time. However, we still need $O(nd^2)$ time to
approximate ${\pi}_i^{\mathrm{mMSE}}$ because they depend on $\|{\mathbf{M}}_X^{-1}{\mathbf{x}}_i\|$ for $i=1,2,...,n$.  Based
on aforementioned discussions, the time complexity of Algorithm~2 with
$\boldsymbol{\pi}^{\mathrm{mVc}}$ is $O(nd+\zeta_0r_0d^2+{\zeta}rd^2)$, and the time complexity of Algorithm~2 with $\boldsymbol{\pi}^{\mathrm{mMSE}}$ is $O(nd^2+\zeta_0r_0d^2+{\zeta}rd^2)$. Considering the case of a very
large $n$ such that $d$, $\zeta_0$, ${\zeta}$, $r_0$ and ${r}$ are all
much smaller than $n$, these time complexities are $O(nd)$ and
$O(nd^2)$, respectively.

  \subsection{Asymptotic properties}
  \label{sec:asympt-prop}
 For the estimator obtained from Algorithm 2 based on the SSPs
  $\boldsymbol{\pi}^{\mathrm{mVc}}$, we derive its asymptotic properties
  under the following assumption.
  \begin{assumption}\label{asp:4}
    The covariate distribution satisfies that ${\mathrm{E}}({\mathbf{x}}{\mathbf{x}}^T)$ is positive definite 
    and ${\mathrm{E}}(e^{\mathbf{a}^T{\mathbf{x}}})<\infty$ for any
    $\mathbf{a}\in\mathbb{R}^d$.
  \end{assumption}
  Assumption~\ref{asp:4} imposes two conditions on covariate
  distribution. The first condition ensures that the asymptotic
  covariance matrix is full rank. 
  The second condition requires that covariate distributions have
  light tails. Clearly, the class of sub-Gaussian distributions
  \citep{buldygin1980sub} satisfy this condition. The main result in
  \cite{owen2007infinitely} also requires this condition.

  We establish the consistency and asymptotic normality of
  ${\breve{{\boldsymbol{\beta}}}}$ based on $\boldsymbol{\pi}^{\mathrm{mVc}}$. The results are
  presented in the following two theorems.
  \begin{theorem}\label{thm:5}
    Let $r_0r^{-1/2}\rightarrow0$. Under Assumption~\ref{asp:4}, if the estimate
    $\tilde{{\boldsymbol{\beta}}}_0$ based on the first step sample exists, then,
    as ${r}\rightarrow\infty$ and $n\rightarrow\infty$, with
    probability approaching one, for any $\epsilon>0$, there exists a
    finite $\Delta_\epsilon$ and $r_\epsilon$ such that
  \begin{equation*}
  P(\|{\breve{{\boldsymbol{\beta}}}}-\hat{{\boldsymbol{\beta}}}_{{{\textnormal{\tiny MLE}}}}\|\ge
  {r}^{-1/2}\Delta_\epsilon|{\mathcal{F}_n})<\epsilon
\end{equation*}
for all ${r}>r_\epsilon$.
\end{theorem}
  In Theorem~\ref{thm:5}, as long as the first step sample estimate $\tilde{{\boldsymbol{\beta}}}_0$ exist, the two step algorithm produces a consistent estimator. We do not even require that $r_0\rightarrow\infty$. If the first step subsample $r_0\rightarrow\infty$, then from Theorem~\ref{thm:1}, $\tilde{{\boldsymbol{\beta}}}_0$ exists with probability approaching one. Under this scenario, the resultant two-step estimator is optimal in the sense of Theorem~\ref{thm:4}. We present this result in the following theorem.

\begin{theorem}\label{thm:6}
  Assume that $r_0r^{-1/2}\rightarrow0$. Under Assumption \ref{asp:4}, as $r_0\rightarrow\infty$,
  ${r}\rightarrow\infty$, and $n\rightarrow\infty$, conditional on
  ${\mathcal{F}_n}$ and $\tilde{{\boldsymbol{\beta}}}_0$,
  \begin{equation*}
    {\mathbf{V}}^{-1/2}({\breve{{\boldsymbol{\beta}}}}-\hat{{\boldsymbol{\beta}}}_{{{\textnormal{\tiny MLE}}}})
    \longrightarrow N(0,\mathbf{I})
  \end{equation*}
  in distribution, in which ${\mathbf{V}}={\mathbf{M}}_X^{-1}{\mathbf{V}}_c{\mathbf{M}}_X^{-1}$ with ${\mathbf{V}}_c$ having the expression of
  \begin{equation}\label{eq:52}
    {\mathbf{V}}_c=\frac{1}{rn^2}
    \left\{\sum_{i=1}^n|y_i-p_i(\hat{{\boldsymbol{\beta}}}_{{{\textnormal{\tiny MLE}}}})|\|{\mathbf{x}}_i\|\right\}
    \left\{
\sum_{i=1}^n\frac{|y_i-p_i(\hat{{\boldsymbol{\beta}}}_{{{\textnormal{\tiny MLE}}}})|{\mathbf{x}}_i{\mathbf{x}}_i^T}{\|{\mathbf{x}}_i\|}\right\}.
\end{equation}
\end{theorem}
\begin{remark}
  In Theorem~\ref{thm:6}, we require that $r_0\rightarrow\infty$ to
  get a consistent pilot estimate which is used to identify the more
  informative data points in the second step, but $r_0$ should be much
  smaller than $r$ so that the more informative second step subsample
  dominates the likelihood function.
\end{remark}
  Theorem~\ref{thm:6} shows that the two-step algorithm is asymptotically more efficient than the uniform subsampling or the case-control subsampling in the sense of Theorem~\ref{thm:4}. From Theorem~\ref{thm:2}, as $r_0\rightarrow\infty$, $\tilde{{\boldsymbol{\beta}}}_0$ is also asymptotic normal, but from Theorem~\ref{thm:4}, the value of $\mathrm{tr}({\mathbf{V}}_c)$ for its asymptotic variance is larger than that for \eqref{eq:52} with the same total subsample sizes.

\subsection{Standard error formula}
  As pointed out by a referee, the standard error of an estimator is also important and needs to be estimated. It is crucial for statistical inferences such as hypothesis testing and confidence interval construction. The asymptotic normality in Theorems \ref{thm:2} and \ref{thm:6} can be used to construct formulas to estimate the standard error. A simple way is to replace $\hat{{\boldsymbol{\beta}}}_{{{\textnormal{\tiny MLE}}}}$ with $\breve{{\boldsymbol{\beta}}}$ in the asymptotic variance-covariance matrix in Theorem \ref{thm:2} or \ref{thm:6} to get the estimated version. This approach, however, requires calculations on the full data. We give a formula that involves only the selected subsample to estimate the variance-covariance matrix.

We propose to estimate the variance-covariance matrix of $\breve{{\boldsymbol{\beta}}}$ using
 \begin{equation}\label{eq:24}
   \breve{{\mathbf{V}}}=\breve{{\mathbf{M}}}_X^{-1}\breve{{\mathbf{V}}}_c\breve{{\mathbf{M}}}_X^{-1},
 \end{equation}
where
\begin{equation*}
\breve{{\mathbf{M}}}_X=\frac{1}{n(r_0+r)}\sum_{i=1}^{r_0+r} \frac{w_i^*(\breve{{\boldsymbol{\beta}}}){\mathbf{x}}_i^*({\mathbf{x}}_i)^{*T}}{\pi_i^*},
\end{equation*}
and
  \begin{equation*}
    \breve{{\mathbf{V}}}_c=\frac{1}{n^2(r_0+r)^2}\sum_{i=1}^{r_0+r}
    \frac{\{y_i^*-p_i^*(\breve{{\boldsymbol{\beta}}})\}^2{\mathbf{x}}_i^*({\mathbf{x}}_i^*)^T}{(\pi_i^*)^2}.
  \end{equation*}
  In the above formula, $\breve{{\mathbf{M}}}_X$ and $\breve{{\mathbf{V}}}_c$ are motivated by the method of moments. If $\breve{{\boldsymbol{\beta}}}$ is replace by $\hat{{\boldsymbol{\beta}}}_{{{\textnormal{\tiny MLE}}}}$, then $\breve{{\mathbf{M}}}_X$ and $\breve{{\mathbf{V}}}_c$ are unbiased estimators of ${\mathbf{M}}_X$ and ${\mathbf{V}}_c$, respectively. Standard errors of components of $\breve{{\boldsymbol{\beta}}}$ can be estimated by the square roots of the diagonal elements of $\breve{{\mathbf{V}}}$. We will evaluate the performance of the formula in \eqref{eq:24} using numerical experiments in Section~\ref{sec:numerical-examples}.

\section{Numerical examples}\label{sec:numerical-examples}
We evaluate the performance of the OSMAC approach using synthetic and
real data sets in this section. We have some additional numerical results in Section~\ref{sec:addit-numer-results} of the Supplementary Material, in which Section~\ref{sec:furth-numer-eval} presents additional results of the OSMAC approach on rare event data and Section~\ref{sec:numer-results-uncond} gives unconditional results. As shown in Theorem \ref{thm:1},
  the approximation error can be arbitrarily small when the subsample
  size gets large enough, so any level of accuracy can be achieved
  even using uniform subsampling as long as the subsample size is
  sufficiently large. In order to make fair comparisons with uniform
  subsampling, we set the total subsample sizes for a two-step procedure the
  same as that for the uniform subsampling approach. In the second step
  of all two-step procedures, except the local case-control (LCC) procedure, we combine the two-step subsamples in
  estimation. This is valid for the OSMAC approach. However, for the LCC procedure, the first step subsample cannot be combined and only the second step subsample can be used. Otherwise, the resultant estimator will be biased  \citep{fithian2014local}.

\subsection{Simulation experiments}\label{sec:simul-exper}
In this section, we use numerical experiments based on simulated data
sets to evaluate the OSMAC approach proposed in previous
sections. Data of size $n=10,000$ are generated from model~\eqref{eq:1}
with the true value of ${\boldsymbol{\beta}}$, ${\boldsymbol{\beta}}_0$, being a $7\times1$ vector
of 0.5. We consider the following 6 simulated data sets using
different distributions of ${\mathbf{x}}$ (detailed definitions of these
distributions can be found in Appendix A of
\cite{gelman2014bayesian}).
\begin{enumerate}[ 1)]\itemsep-0em
\item {\bf mzNormal}. ${\mathbf{x}}$ follows a multivariate normal distribution
  with mean $\boldsymbol{0}$, $N(\boldsymbol{0},\boldsymbol{\Sigma})$, where
  $\boldsymbol{\Sigma}_{ij}=0.5^{I(i\neq j)}$ and $I()$ is the indicator
  function. For this data set, the number of 1's and the number of 0's
  in the responses are roughly equal. This data set is referred to as
  mzNormal data.
\item {\bf nzNormal}. ${\mathbf{x}}$ follows a multivariate normal distribution
  with nonzero mean, $N(\boldsymbol{1.5},\boldsymbol{\Sigma})$. About 95\% of the
  responses are 1's, so this data set is an example of imbalanced data
  and it is referred to as nzNormal data. 
\item {\bf ueNormal}. ${\mathbf{x}}$ follows a multivariate normal distribution
  with zero mean but its components have unequal variances. To be
  specific, let ${\mathbf{x}}=(x_1,...x_7)^T$, in which $x_i$ follows a normal
  distribution with mean 0 and variance $1/i^2$ and the correlation
  between $x_i$ and $x_j$ is $0.5^{I(i\neq j)}$, $i,j=1,...,7$. For this
  data set, the number of 1's and the number of 0's in the responses
  are roughly equal. This data set is referred to as ueNormal data.
\item {\bf mixNormal}. ${\mathbf{x}}$ is a mixture of two multivariate normal
  distributions with different means, i.e.,
  ${\mathbf{x}}\sim 0.5N(\boldsymbol{1},\boldsymbol{\Sigma})+0.5N(-\boldsymbol{1},\boldsymbol{\Sigma})$. For
  this case, the distribution of ${\mathbf{x}}$ is bimodal, and the number of
  1's and the number of 0's in the responses are roughly equal. This
  data set is referred to as mixNormal data.
\item {$\mathbf{T_3}$}. ${\mathbf{x}}$ follows a multivariate $t$ distribution
  with degrees of freedom 3, $t_3(\boldsymbol{0},\boldsymbol{\Sigma})/10$. For this case, the distribution of ${\mathbf{x}}$ has heavy tails and it does not satisfy the conditions in Sections~\ref{sec:random-subsampling} and \ref{sec:two-step}. We use this case to exam how sensitive the OSMAC approach is to the required assumptions. The number
  of 1's and the number of 0's in the responses are roughly
  equal for this data set. It is referred to as $T_3$ data.
\item {\bf EXP}. Components of ${\mathbf{x}}$ are independent and each has an
  exponential distribution with a rate parameter of 2. For this case,
  the distribution of ${\mathbf{x}}$ is skewed and has a heavier tail on the
  right, and the proportion of 1's in the responses is about
  0.84. This data set is referred to as EXP data.
\end{enumerate}

In order to clearly show the effects of different distributions of
${\mathbf{x}}$ on the SSP, we create boxplots of SSPs, shown in
Figure~\ref{fig:1} for the six data sets. It is seen that
distributions of covariates have high influence on optimal
SSPs. Comparing the figures for the mzNormal and nzNormal data sets,
we see that a change in the mean makes the distributions of SSPs
dramatically different. Another evident pattern is that using ${\mathbf{V}}_c$
instead of ${\mathbf{V}}$ to define an optimality criterion makes the SSP
different, especially for the case of unNormal data set which has
unequal variances for different components of the covariate. For the
mzNormal and $T_3$ data sets, the difference in the SSPs are not
evident. For the EXP data set, there are more points in the two
tails of the distributions.

\begin{figure}[htp]
  \centering
  \begin{subfigure}{0.3\textwidth}
    \includegraphics[width=\textwidth]{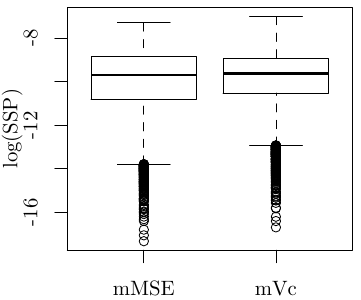}
    \caption{mzNormal}
  \end{subfigure}
  \begin{subfigure}{0.3\textwidth}
    \includegraphics[width=\textwidth]{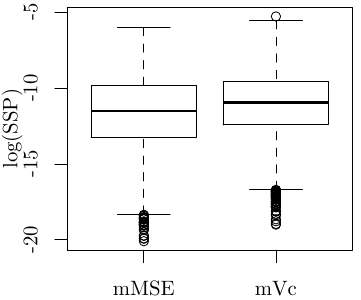}
    \caption{nzNormal}
  \end{subfigure}
  \begin{subfigure}{0.3\textwidth}
    \includegraphics[width=\textwidth]{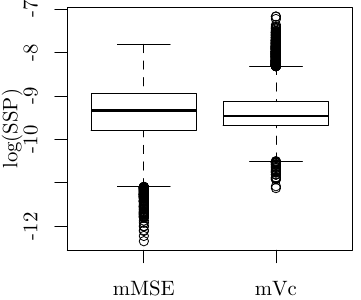}
    \caption{ueNormal}
  \end{subfigure}
  \begin{subfigure}{0.3\textwidth}
    \includegraphics[width=\textwidth]{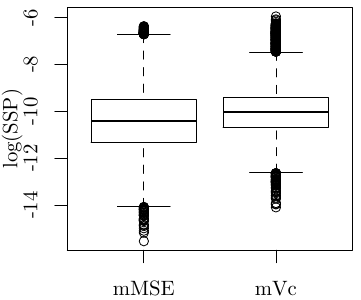}
    \caption{mixNormal}
  \end{subfigure}
  \begin{subfigure}{0.3\textwidth}
    \includegraphics[width=\textwidth]{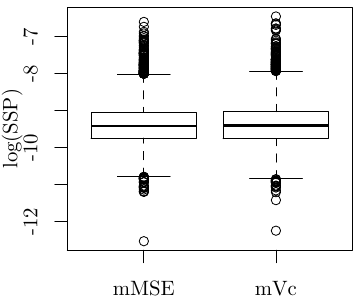}
    \caption{$T_3$.}
  \end{subfigure}
  \begin{subfigure}{0.3\textwidth}
    \includegraphics[width=\textwidth]{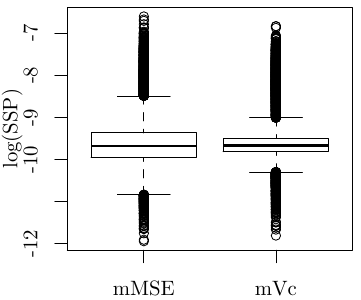}
    \caption{EXP}
  \end{subfigure}\\[-3mm]
  \caption{Boxplots of SSPs for different data sets. Logarithm is
    taken on SSPs for better presentation of the figures.}
  \label{fig:1}
\end{figure}

Now we evaluate the performance of Algorithm 2 based on different
choices of SSPs.  We calculate MSEs of $\breve{{\boldsymbol{\beta}}}$ from $S=1000$
subsamples using
$\textrm{MSE}=S^{-1}\sum_{s=1}^S\|\breve{{\boldsymbol{\beta}}}^{(s)}-\hat{{\boldsymbol{\beta}}}_{{{\textnormal{\tiny MLE}}}}\|^2$,
where $\breve{{\boldsymbol{\beta}}}^{(s)}$ is the estimate from the $s$th subsample.
Figure~\ref{fig:2} presents the MSEs of $\breve{{\boldsymbol{\beta}}}$ from
Algorithm 2 based on different SSPs, where the first step sample size
$r_0$ is fixed at 200. For comparison, we provide the results
of uniform subsampling and the LCC subsampling. We also calculate the full data MLE using 1000 Bootstrap samples.

For all the six data sets, SSPs $\boldsymbol{\pi}^{\mathrm{mMSE}}$ and $\boldsymbol{\pi}^{\mathrm{mVc}}$ always
result in smaller MSE than the uniform SSP, which agrees with the theoretical result
that they aim to minimize the asymptotic MSEs of the resultant estimator.
If components of ${\mathbf{x}}$ have equal variances, the OSMAC with
$\boldsymbol{\pi}^{\mathrm{mMSE}}$ and $\boldsymbol{\pi}^{\mathrm{mVc}}$ have
similar performances; for the ueNormal data set this is not true, and
the OSMAC with $\boldsymbol{\pi}^{\mathrm{mMSE}}$ dominates the OSMAC with
$\boldsymbol{\pi}^{\mathrm{mVc}}$. The uniform SSP never yields the smallest
MSE. It is worth noting that both the two OSMAC methods outperforms
the uniform subsampling method for the $T_3$ and EXP data sets. This
indicates that the OSMAC approach has advantage over the uniform subsampling even when data
do not satisfy the assumptions imposed in
Sections~\ref{sec:appr-optim-subs} and \ref{sec:two-step}. For the LCC subsampling, it can be less efficient than the OSMAC procedure if the data set is not very imbalanced. It performs well  for the nzNormal data which is imbalanced. This agree with the goal of the method in dealing with imbalanced data. The LCC subsampling does not perform well for small $r$. The main reason is that this method cannot use the first step sample so the effective sample size is smaller than other methods. 

\begin{figure}[htp]
  \centering
  \begin{subfigure}{0.49\textwidth}
    \includegraphics[width=\textwidth]{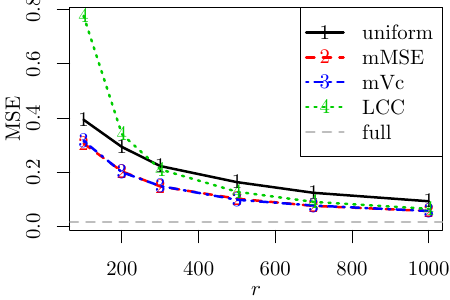}
    \caption{mzNormal}
  \end{subfigure}
  \begin{subfigure}{0.49\textwidth}
    \includegraphics[width=\textwidth]{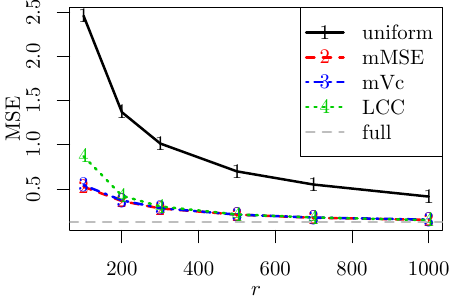}
    \caption{nzNormal}
  \end{subfigure}\\[5mm]
  \begin{subfigure}{0.49\textwidth}
    \includegraphics[width=\textwidth]{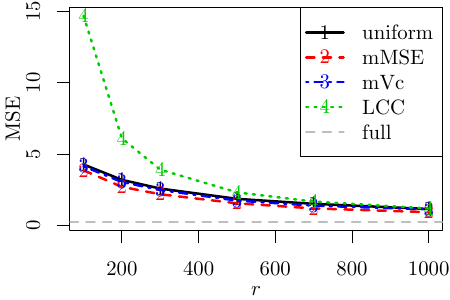}
    \caption{ueNormal}
  \end{subfigure}
  \begin{subfigure}{0.49\textwidth}
    \includegraphics[width=\textwidth]{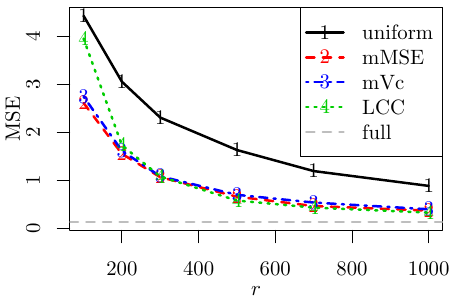}
    \caption{mixNormal}
  \end{subfigure}\\[5mm]
  \begin{subfigure}{0.49\textwidth}
    \includegraphics[width=\textwidth]{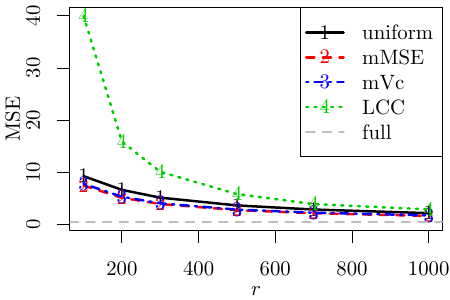}
    \caption{$T_3$}
  \end{subfigure}
  \begin{subfigure}{0.49\textwidth}
    \includegraphics[width=\textwidth]{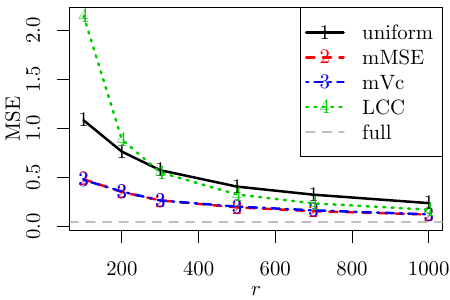}
    \caption{EXP}
  \end{subfigure}\\[5mm]
  \caption{MSEs for different second step subsample size ${r}$ with
    the first step subsample size being fixed at $r_0=200$.}
  \label{fig:2}
\end{figure}

To investigate the effect of different sample size allocations between
the two steps, we calculate MSEs for various
proportions of first step samples with fixed total subsample sizes.
Results are given in Figure~\ref{fig:3} with total subsample size
$r_0+r=800$ and $1200$ for the mzNormal data set. It shows that, the
performance of a two-step algorithm improves at first by increasing
$r_0$, but then it becomes less efficient after a certain point as
$r_0$ gets larger. This is because if $r_0$ is too small, the first
step estimate is not accurate; if $r_0$ is too close to $r$, then the
more informative second step subsample would be small. These
observations indicate that, empirically, a value around 0.2 is a good
choice for $r_0/(r_0+r)$ in order to have an efficient two-step
algorithm. However, finding a systematic way of determining the
optimal sample sizes allocation between two steps needs further
study. Results for the other five data sets are similar so they are
omitted to save space.

\begin{figure}[htp]
  \centering
  \begin{subfigure}{0.45\textwidth}
    \includegraphics[width=\textwidth]{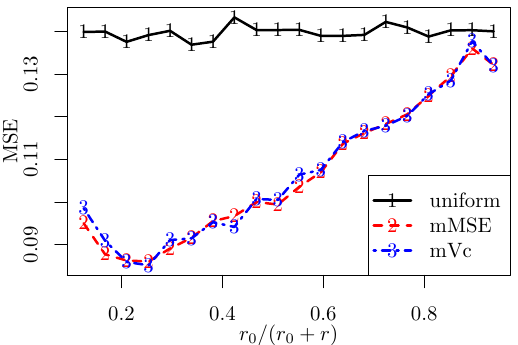}
    \caption{$r_0+r=800$}
  \end{subfigure}
  \begin{subfigure}{0.45\textwidth}
    \includegraphics[width=\textwidth]{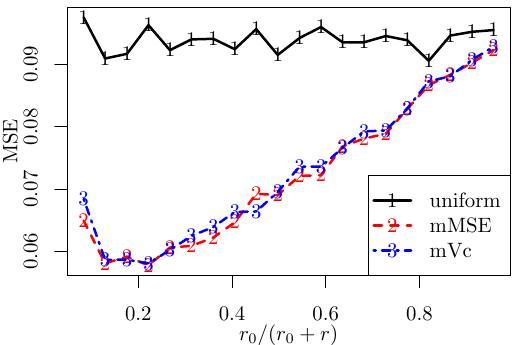}
    \caption{$r_0+r=1200$}
  \end{subfigure}
  \caption{MSEs vs proportions of the first step subsample with fixed
    total subsample sizes for the mzNormal data set.}
  \label{fig:3}
\end{figure}

 Figure~\ref{fig:4} gives proportions of correct classifications on the
  responses using different methods. To avoid producing
  over-optimistic results, we generate two full data sets
  corresponding to each of the six scenarios, use one of them to
  obtain estimates with different methods, and then perform
  classification on the other full data. The
classification rule is to classify the response to be 1 if
$p_i(\breve{{\boldsymbol{\beta}}})$ is larger than 0.5, and 0 otherwise. For
comparisons, we also use the full data MLE to classify the
full data. As shown in Figure~\ref{fig:4}, all the methods, except LCC with small $r$, produce
proportions close to that from using the full data MLE, showing the
comparable performance of the OSMAC algorithms to that of the full
data approach in classification. 

\begin{figure}[htp]
  \centering
  \begin{subfigure}{0.49\textwidth}
    \includegraphics[width=\textwidth]{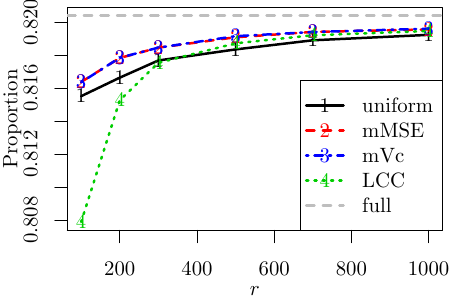}
    \caption{mzNormal}
  \end{subfigure}
  \begin{subfigure}{0.49\textwidth}
    \includegraphics[width=\textwidth]{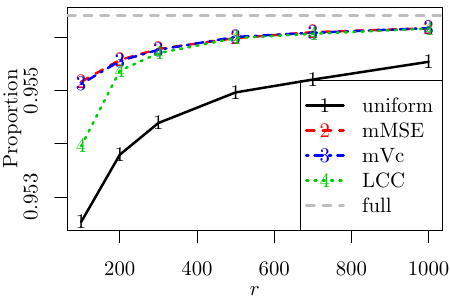}
    \caption{nzNormal}
  \end{subfigure}\\[5mm]
  \begin{subfigure}{0.49\textwidth}
    \includegraphics[width=\textwidth]{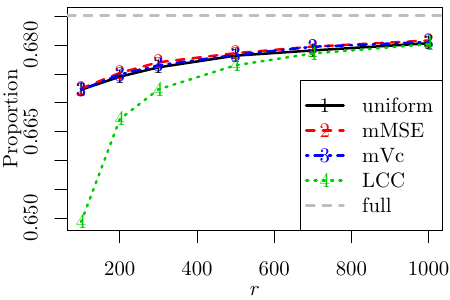}
    \caption{ueNormal}
  \end{subfigure}
  \begin{subfigure}{0.49\textwidth}
    \includegraphics[width=\textwidth]{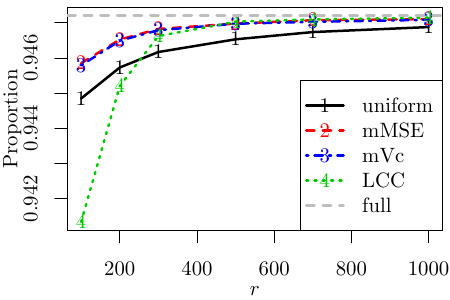}
    \caption{mixNormal}
  \end{subfigure}\\[5mm]
  \begin{subfigure}{0.49\textwidth}
    \includegraphics[width=\textwidth]{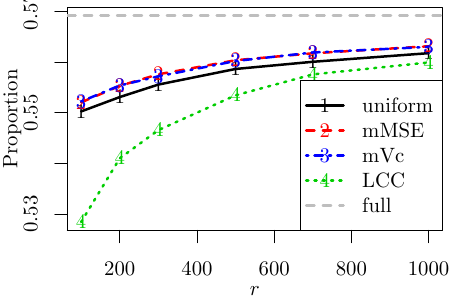}
    \caption{$T_3$}
  \end{subfigure}
  \begin{subfigure}{0.49\textwidth}
    \includegraphics[width=\textwidth]{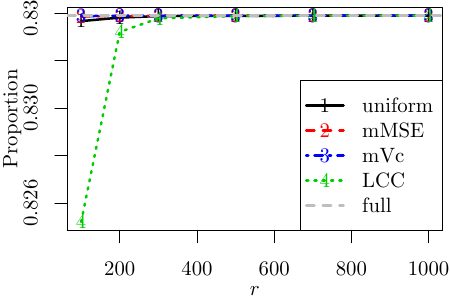}
    \caption{EXP}
  \end{subfigure}\\[5mm]
  \caption{Proportions of correct classifications for different second
    step subsample size ${r}$ with the first step subsample size
    being fixed at $r_0=200$. The gray horizontal dashed lines are those using the true parameter.}
  \label{fig:4}
\end{figure}

 To assess the performance of the formula in \eqref{eq:24}, we
  use it to calculate the estimated MSE, i.e., $\mathrm{tr}(\tilde{{\mathbf{V}}})$, and
  compare the average estimated MSE with the empirical 
  MSE. Figure~\ref{fig:a} presents the results for OSMAC with
  $\boldsymbol{\pi}^{\mathrm{mMSE}}$. It is seen that the estimated MSEs are
  very close to the empirical MSEs, except for the case of nzNormal
  data which is imbalanced. This indicates that the proposed formula
  works well if the data is not very imbalanced. According to our 
  simulation experiments, it works well if the proportion of 1's in
  the responses is between 0.15 and 0.85. For more imbalanced data or
  rare events data, the formula may not be accurate because the
  properties of the MLE are different from these for the regular cases
  \citep{owen2007infinitely, king2001logistic}. The performance of the formula in \eqref{eq:24}
  for OSMAC with $\boldsymbol{\pi}^{\mathrm{mVc}}$ is
  similar to that for OSMAC with $\boldsymbol{\pi}^{\mathrm{mMSE}}$, so
  results are omitted for clear presentation of the plot.

  To further evaluate the performance of the proposed method in
  statistical inference, we consider confidence interval construction
  using the asymptotic normality and the estimated variance-covariance matrix
  in \eqref{eq:24}. For illustration, we take the parameter of
  interest as $\beta_1$, the first element of ${\boldsymbol{\beta}}$. The corresponding 95\% confidence interval is
  constructed using
  $\breve{\beta}_1\pm Z_{0.975}SE_{\breve{\beta}_1}$, where
  $SE_{\breve{\beta}_1}=\sqrt{\breve{V}_{11}}$ is the standard error
  of $\breve{\beta}_1$, and $Z_{0.975}$ is the 97.5th percentile of
  the standard normal distribution. We repeat the simulation 3000
  times and estimate the coverage probability of the confidence
  interval by the proportion that it covers the true vale of
  $\beta_1$. Figure~\ref{fig:b} gives the results. The confidence
  interval works perfectly for the mxNormal, ueNormal and $T_3$
  data. For mixNormal and EXP data sets, the empirical coverage
  probabilities are slightly smaller than the intended confidence
  level, but the results are acceptable. For the imbalanced nzNormal
  data, the coverage probabilities are lower than the nominal coverage probabilities. This agrees with
  the fact in Figure~\ref{fig:a} that the formula in \eqref{eq:24}
  does not approximate the asymptotic variance-covariance matrix well for imbalance
  data.

\begin{figure}[htp]
  \centering
  \begin{subfigure}{0.49\textwidth}
    \includegraphics[width=\textwidth]{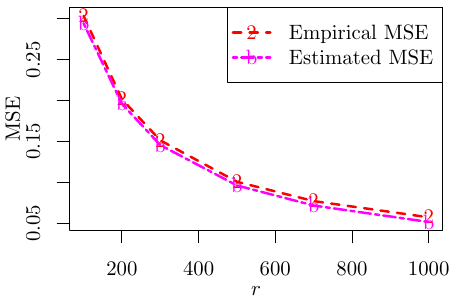}
    \caption{mzNormal}
  \end{subfigure}
  \begin{subfigure}{0.49\textwidth}
    \includegraphics[width=\textwidth]{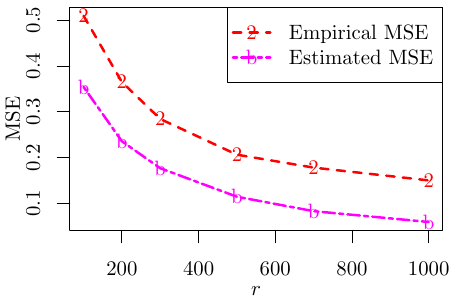}
    \caption{nzNormal}
  \end{subfigure}\\[5mm]
  \begin{subfigure}{0.49\textwidth}
    \includegraphics[width=\textwidth]{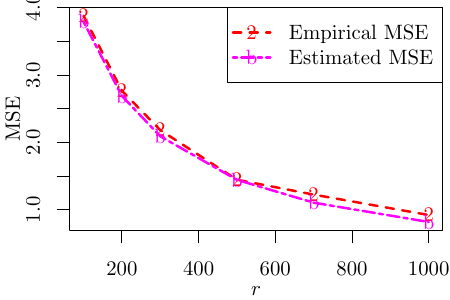}
    \caption{ueNormal}
  \end{subfigure}
  \begin{subfigure}{0.49\textwidth}
    \includegraphics[width=\textwidth]{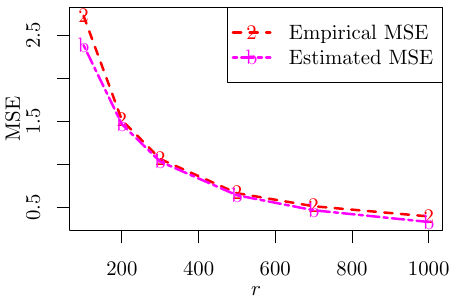}
    \caption{mixNormal}
  \end{subfigure}\\[5mm]
  \begin{subfigure}{0.49\textwidth}
    \includegraphics[width=\textwidth]{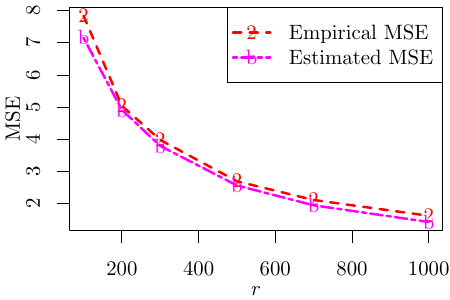}
    \caption{$T_3$}
  \end{subfigure}
  \begin{subfigure}{0.49\textwidth}
    \includegraphics[width=\textwidth]{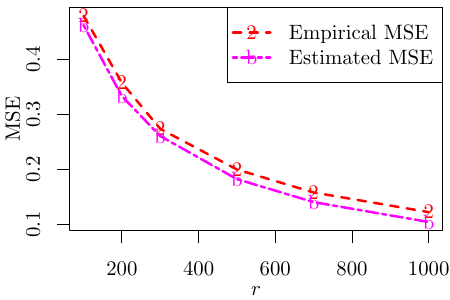}
    \caption{EXP}
  \end{subfigure}\\[5mm]
  \caption{Estimated and empirical MSEs for the OSMAC with $\boldsymbol{\pi}^{\mathrm{mMSE}}$. The first step subsample size is fixed at $r_0=200$ and the second step subsample size ${r}$ changes.}
  \label{fig:a}
\end{figure}

\begin{figure}[htp]
  \centering
  \begin{subfigure}{0.49\textwidth}
    \includegraphics[width=\textwidth]{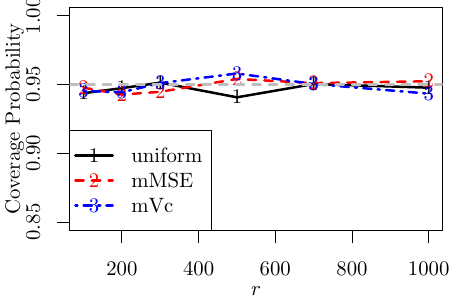}
    \caption{mzNormal}
  \end{subfigure}
  \begin{subfigure}{0.49\textwidth}
    \includegraphics[width=\textwidth]{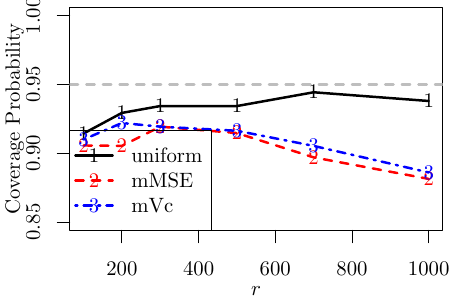}
    \caption{nzNormal}
  \end{subfigure}\\[5mm]
  \begin{subfigure}{0.49\textwidth}
    \includegraphics[width=\textwidth]{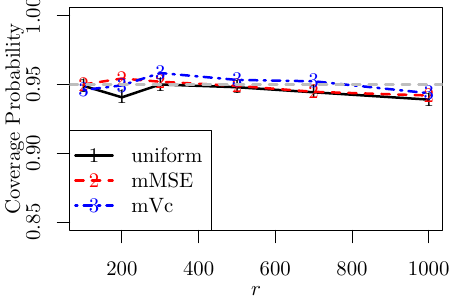}
    \caption{ueNormal}
  \end{subfigure}
  \begin{subfigure}{0.49\textwidth}
    \includegraphics[width=\textwidth]{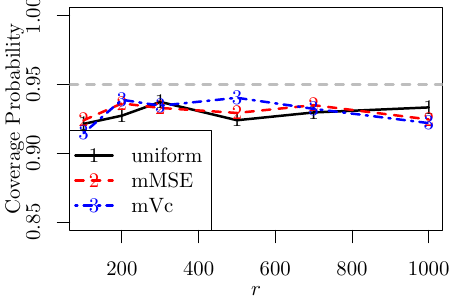}
    \caption{mixNormal}
  \end{subfigure}\\[5mm]
  \begin{subfigure}{0.49\textwidth}
    \includegraphics[width=\textwidth]{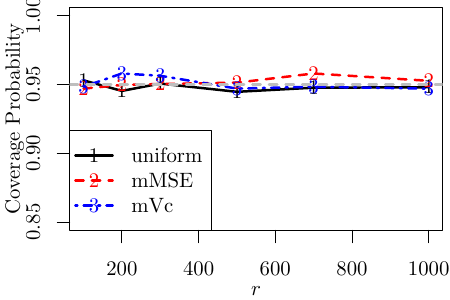}
    \caption{$T_3$}
  \end{subfigure}
  \begin{subfigure}{0.49\textwidth}
    \includegraphics[width=\textwidth]{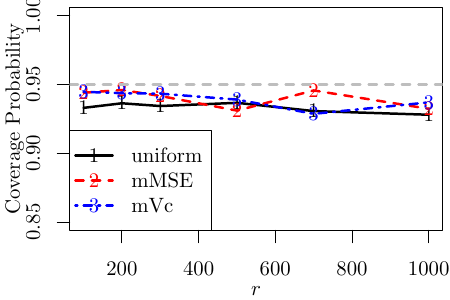}
    \caption{EXP}
  \end{subfigure}\\[5mm]
  \caption{Empirical coverage probabilities for different second step subsample size r with the first step subsample
size being fixed at $r_0=200$.}
  \label{fig:b}
\end{figure}

To evaluate the computational efficiency of the subsampling
algorithms, we record the computing time and numbers of iterations of
Algorithm 2 and the uniform subsampling implemented in the R
programming language \citep{R}. Computations were carried out on a
desktop running Window 10 with an Intel I7 processor and 16GB
memory.  For fair comparison, we counted only the CPU time used by
1000 repetitions of each method. Table~\ref{tab:1} gives the results
for the mzNormal data set for algorithms based on
$\boldsymbol{\pi}^{\mathrm{mMSE}}$, $\boldsymbol{\pi}^{\mathrm{mVc}}$, and
$\boldsymbol{\pi}^{\mathrm{UNI}}$. The computing time for using the full data
is also given in the last row of Table~\ref{tab:1} for comparisons. It
is not surprising to observe  that the uniform subsampling algorithm requires the
least computing time because it does not require an additional step to
calculate the SSP. The algorithm based on $\boldsymbol{\pi}^{\mathrm{mMSE}}$
requires longer computing time than the algorithm based on
$\boldsymbol{\pi}^{\mathrm{mVc}}$, which agrees with the theoretical analysis
in Section~\ref{sec:two-step}. All the subsampling algorithms take
significantly less computing time compared to using the full data
approach. Table~\ref{tab:2} presents the average numbers of iterations
in Newton's method. It shows that for Algorithm 2, the first step may
require additional iterations compared to the second step, but
overall, the required numbers of iterations for all methods are close
to 7, the number of iterations used by the full data. This shows that
using a smaller subsample does not increase the required number of
iterations much for Newton's method.

\begin{table}[htp]
  \caption{CPU seconds for the mzNormal data set with $r_0=200$ and different ${r}$. The CPU seconds for using the full data is given in the last row.}
\centering
  \begin{tabular}{ccccccc}\hline
    Method & \multicolumn{6}{c}{$r$} \\
    \cline{2-7}
           & 100 & 200 & 300 & 500 & 700 & 1000 \\ \hline
    mMSE & 3.340 & 3.510 & 3.720 & 4.100 & 4.420 & 4.900 \\
    mVc & 3.000 & 3.130 & 3.330 & 3.680 & 4.080 & 4.580 \\
    Uniform & 0.690 & 0.810 & 0.940 & 1.190 & 1.470 & 1.860 \\  \hline
    \multicolumn{7}{l}{Full data CPU seconds: 13.480}\\\hline
  \end{tabular}
  \label{tab:1}
\end{table}

\begin{table}[htp]
  \centering
  \caption{Average numbers of iterations used in Newton's method~\eqref{eq:21} for the mzNormal data set with $r_0=200$ and different ${r}$. For the full data, the number of iterations is 7.}
  \begin{tabular}{lccccccc}\hline
    ${r}$ & \multicolumn{2}{c}{mMSE} &  & \multicolumn{2}{c}{mVc} &  & uniform \\
    \cline{2-3}\cline{5-6}\cline{8-8}
          & First step & Second step && First step & Second step && \\
    \hline
    100 & 7.479 & 7.288 &  & 7.479 & 7.296 &  & 7.378 \\
    200 & 7.479 & 7.244 &  & 7.479 & 7.241 &  & 7.305 \\
    300 & 7.482 & 7.230 &  & 7.482 & 7.214 &  & 7.259 \\
    500 & 7.458 & 7.200 &  & 7.458 & 7.185 &  & 7.174 \\
    700 & 7.472 & 7.190 &  & 7.472 & 7.180 &  & 7.136 \\
    1000 & 7.471 & 7.181 &  & 7.471 & 7.158 &  & 7.091 \\
    \hline
  \end{tabular}\\
  \label{tab:2}
\end{table}

To further investigate the computational gain of the subsampling approach for massive data volume, we increase the value of $d$ to $d=50$ and increase the values of $n$ to be $n=10^4, 10^5, 10^6$ and $10^7$. We record the computing time for the case when ${\mathbf{x}}$ is multivariate normal. Table~\ref{tab:6} presents the result based on one iteration of calculation. It is seen that as $n$ increases, the computational efficiency for a subsampling method relative to the full data approach is getting more and more significant.

\begin{table}[htp]
  \caption{CPU seconds with $r_0=200$, ${r}=1000$ and different full data size $n$ when the covariates are from a $d=50$ dimensional normal distribution.}
\centering
\begin{tabular}{ccccc}\hline
  Method & \multicolumn{ 4}{c}{$n$} \\
  \cline{2-5}
      & $10^4$ & $10^5$ & $10^6$ & $10^7$ \\ \hline
  mMSE & 0.050 & 0.270 & 3.290 & 37.730 \\
  mVc & 0.030 & 0.070 & 0.520 & 6.640 \\
  Uniform & 0.010 & 0.030 & 0.020 & 0.830 \\
  Full & 0.150 & 1.710 & 16.530 & 310.450 \\ \hline
\end{tabular}
\label{tab:6}
\end{table}

\subsubsection{Numerical evaluations for rare events data}
\label{sec:numer-eval-rare}

To investigate the performance of the proposed method for the case of rare events, we generate rare events data using the same configurations that are used to generate the nzNormal data, except that we change the mean of ${\mathbf{x}}$ to {\bf -2.14} or {\bf -2.9}. With these values, 1.01\% and 0.14\% of responses are 1 in the full data of size $n=10000$.

Figure~\ref{fig:s1} presents the results for these two scenarios. It is seen that both mMSE and mVc work well for these two scenarios and their performances are similar. The uniform subsampling is neither stable nor efficient. When the event rate is 0.14\%, corresponding to the subsample sizes of 300, 400, 500, 700, 900, and 1200, there are 903, 848, 801, 711, 615, and 491 cases out of 1000 repetitions of the simulation that the MLE are not found. For the cases that the MLE are found, the MSEs are 78.27907, 23.28546, 34.16891, 42.43081, 26.38999, and 19.25178, respectively. These MSEs are much larger than those from the OSMAC and thus are omitted in Figure~\ref{fig:s1} for better presentation. For the OSMAC, there are 8 cases out of 1000 that the MLE are not found only when $r_0=200$ and $r=100$.

For comparison, we also calculate the MSE of the full data approach
using 1000 Bootstrap samples (the gray dashed line). Note that the Bootstrap is the uniform subsampling with the subsample size being equal to the full data sample size. Interestingly, it is seen from Figure~\ref{fig:s1} that OSMAC methods can produce MSEs that are much smaller than the Bootstrap MSEs. To further investigate this interesting case, we carry out another simulation using the exact same setup. A full data is generated in each repetition and hence the resultant MSEs are the unconditional MSEs. Results are presented in Figure~\ref{fig:s2}. Although the unconditional MSEs of the OSMAC methods are larger than that of the full data approach, they are very close when $r$ gets large, especially when the rare event rate is $0.11\%$. Here, $0.11\%$ is the average percentage of 1's in the responses of all 1000 simulated full data. 
Note that the true value of ${\boldsymbol{\beta}}$ is used in calculating both the
conditional MSEs and the unconditional MSEs. Comparing
Figure~\ref{fig:s1} (b) and Figure~\ref{fig:s2} (b), conditional
inference of OSMAC can indeed be more efficient than the full data
approach for rare events data. These two figures also indicate that
the original Bootstrap method does not work perfectly for the case of
rare events data. For additional results on more extreme rare events data, please read Section~\ref{sec:furth-numer-eval} in the Supplementary Material.
\begin{figure}[htp]
  \centering
  \begin{subfigure}{0.49\textwidth}
    \includegraphics[width=\textwidth]{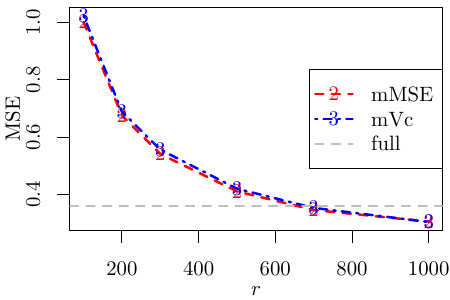}
    \caption{1.01\% of $y_i$'s are 1}
  \end{subfigure}
  \begin{subfigure}{0.49\textwidth}
    \includegraphics[width=\textwidth]{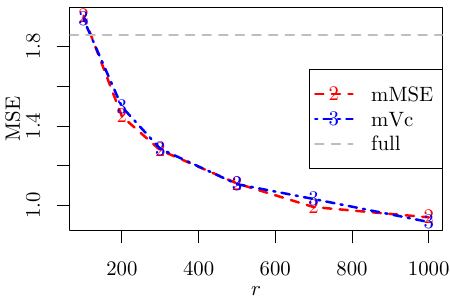}
    \caption{0.14\% of $y_i$'s are 1}
  \end{subfigure}
  \caption{MSEs for rare event data with different second step subsample size $r$ and a fixed first step subsample size $r_0=200$, where the covariates follow multivariate normal distributions.}
  \label{fig:s1}
\end{figure}

\begin{figure}[htp]
  \centering
  \begin{subfigure}{0.49\textwidth}
    \includegraphics[width=\textwidth]{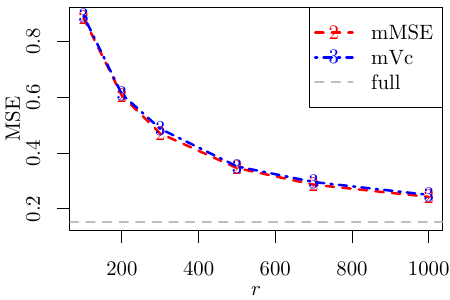}
    \caption{1.04\% of $y_i$'s are 1}
  \end{subfigure}
  \begin{subfigure}{0.49\textwidth}
    \includegraphics[width=\textwidth]{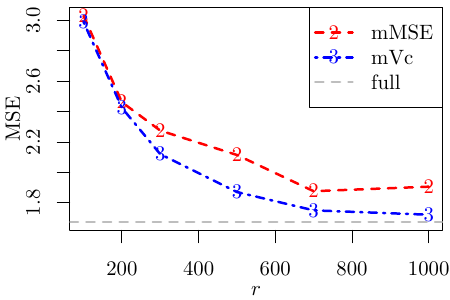}
    \caption{0.11\% of $y_i$'s are 1}
  \end{subfigure}
  \caption{Unconditional MSEs for rare event data with different second step subsample size $r$ and a fixed first step subsample size $r_0=200$, where the covariates follow multivariate normal distributions.}
  \label{fig:s2}
\end{figure}

\subsection{Census income data set}
In this section, we apply the proposed methods to a census income data
set \citep{kohavi-nbtree}, which was extracted from the 1994 Census
database. There are totally $48,842$ observations in this data set, and the response variable is whether a person's income exceeds \$50K a
year. There are 11,687 individuals (23.93\%) in the data whose income exceed \$50K a year. Inferential task is to estimate the effect on income from the
following covariates: $x_1$, age; $x_2$, final weight (Fnlwgt); $x_3$, highest level of education in
numerical form; $x_4$, capital loss (LosCap); $x_5$, hours worked per
week. 
 The variable final weight ($x_2$) is the number of people the observation represents. The values were assigned by Population Division at the Census Bureau, and they are related to the
  socio-economic characteristic, i.e., people with similar socio-economic characteristics have similar weights. Capital loss ($x_5$) is the loss in income due to bad investments; it is the difference between lower selling prices of investments and higher purchasing prices of investments made by the individual. 
  
 The parameter corresponding to $x_i$ is denoted as $\beta_i$ for
$i=1, ..., 5$. An intercept parameter, say $\beta_0$, is also include
in the model.  Another interest is to determine whether a person's
income exceeds \$50K a year using the covariates. We obtained the data from the Machine Learning Repository \citep{Lichman:2013}, where it is partitioned into a training set of $n=32,561$ observations and a validation set of $16,281$ observations. Thus we apply the proposed method on the train set and use the validation set to evaluate the performance of classification.

For this
data set, the full data estimates using all the observation in the
training set are: $\hat\beta_0=-8.637$ (0.116),
$\hat\beta_1=0.637$ (0.016), $\hat\beta_2=0.065$ (0.015),
$\hat\beta_3=0.878$ (0.017),
$\hat\beta_4=0.234$ (0.013) and $\hat\beta_5=0.525$ (0.016), where the
numbers in the parentheses are the associated standard errors. Table~\ref{tab:3} gives the average of parameter estimates along with
the empirical and estimated standard errors from different methods based on
1000 subsamples of $r_0+r=1200$ with $r_0=200$ and ${r}=1000$. It is seen that all
subsampling method produce estimates close to those from the full
data approach. In general, OSMAC with $\boldsymbol{\pi}^{\mathrm{mMSE}}$ and OSMAC
with $\boldsymbol{\pi}^{\mathrm{mVc}}$ produce the smallest standard
errors. The estimated standard errors are very close to the empirical
standard errors, showing that the proposed asymptotic variance-covariance formula in
\eqref{eq:24} works well for the read data.
The standard errors for the subsample estimates are larger than
those for the full data estimates. However, they are quite good in view
of the relatively small subsample size.  
 All methods show that the effect of each variable on income is
 positive. However, the effect of final weight is not significant at significance level 0.05 according to any subsample-based method, while this variable is significant at the same significance level according to the full data analysis. The reason is that the subsample inference is not as powerful as the full data approach due to its relatively smaller sample size. Actually, for statistical inference in large sample, no matter how small the true parameter is, as long as it is a nonzero constant, the corresponding variable can always be detected as significant with large enough sample size. This is also true for conditional inference based on a subsample if the subsample size is large enough. It is interesting that capital loss has a significantly positive effect on income, this is because people with low income seldom have investments.

\begin{table}[htp]
  \caption{Average estimates for the Adult income data set based on
    1000 subsamples. The numbers in the parentheses are the associated
    empirical and average estimated standard errors, respectively. In the table, $\beta_1$ is for age, $\beta_2$ is for final weight, $\beta_3$ is for highest level of education in numerical form, $\beta_4$ is for capital loss, and $\beta_5$ is for hours worked per week.}
\begin{center}
\begin{tabular}{cccc}\hline
 & uniform        & mMSE           & mVc            \\ \hline
 Intercept & -8.686  (0.629, 0.609) & -8.660  (0.430, 0.428) & -8.639  (0.513, 0.510) \\
 $\beta_1$       & 0.638  (0.079, 0.078) & 0.640  (0.068, 0.071) & 0.640  (0.068, 0.067) \\
 $\beta_2$    & 0.061  (0.076, 0.077) & 0.065  (0.067, 0.068) & 0.063  (0.061, 0.062) \\
 $\beta_3$ & 0.882  (0.090, 0.090) & 0.881  (0.079, 0.075) & 0.878  (0.072, 0.072) \\
 $\beta_4$    & 0.232  (0.070, 0.071) & 0.231  (0.058, 0.059) & 0.232  (0.060, 0.057) \\
 $\beta_5$     & 0.533  (0.085, 0.087) & 0.526  (0.068, 0.070) & 0.526  (0.071, 0.070) \\
 \hline
  \end{tabular}
\end{center}
\label{tab:3}
\end{table}

Figure~\ref{fig:5-6} (a) shows the MSEs that were calculated from $S=1000$
subsamples of size $r_0+r$ with a fixed $r_0=200$.  In this figure, all
MSEs are small and go to 0 as the subsample size gets large, showing
the estimation consistency of the subsampling methods. The OSMAC with
$\boldsymbol{\pi}^{\mathrm{mMSE}}$ always has the smallest MSE.
Figure~\ref{fig:5-6} (b) gives the proportions of correct classifications on
the responses in the validation set for different second step subsample sizes with a fixed
$r_0=200$ when the classification threshold is 0.5.  For comparison,
we also obtained the results of classification using the full data estimate which is the gray horizontal dashed line. Indeed, using all
the $n=32,561$ observations in the training set yields better results than using
subsamples of much smaller sizes, but the difference is really
small. One point worth to mention is that although the OSMAC with
$\boldsymbol{\pi}^{\mathrm{mMSE}}$ always yields a smaller MSE compared to
the OSMAC with $\boldsymbol{\pi}^{\mathrm{mVc}}$, its performance in
classification is inferior to the OSMAC with
$\boldsymbol{\pi}^{\mathrm{mVc}}$. This is because
$\boldsymbol{\pi}^{\mathrm{mMSE}}$ aims to minimize the asymptotic MSE and
may not minimize the misclassification rate, although the two goals
are highly related.

\begin{figure}[htp]
  \centering
  \begin{subfigure}{0.48\textwidth}
    \includegraphics[width=\textwidth]{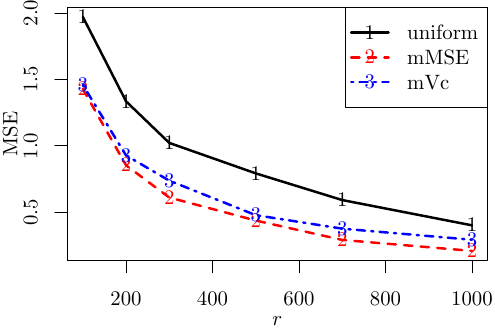}
    \caption{MSEs vs  ${r}$}
  \end{subfigure}
  \begin{subfigure}{0.51\textwidth}
    \includegraphics[width=\textwidth]{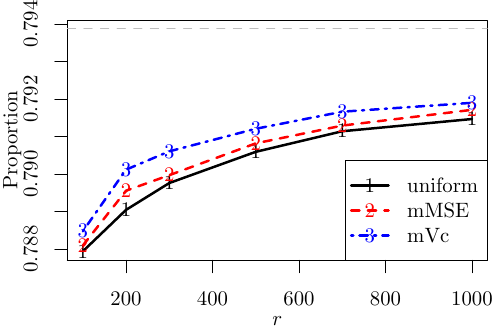}
    \caption{Proportions of correct classifications vs $r$}
  \end{subfigure}  
  \caption{MSEs and proportions of correct classifications for the adult income data set with $r_0=200$ and different second step subsample size ${r}$. The gray horizontal dashed line in figure (b) is the result using the full data MLE.}
  \label{fig:5-6}
\end{figure}

\subsection{Supersymmetric benchmark data set}
We apply the subsampling methods to a supersymmetric (SUSY) benchmark
data set \citep{Baldi2014} in this section. The data set is available from the Machine Learning Repository \citep{Lichman:2013} at this link: \url{https://archive.ics.uci.edu/ml/datasets/SUSY}. The goal is to distinguish
between a process where new supersymmetric particles are produced and
a background process, utilizing the 18 kinematic features in the data set.  The
full sample size is $5,000,000$ and the data file is about 2.4
gigabytes. About 54.24\% of the responses in the full data are from the background process. We use the first $n=4,500,000$ observation as the training set and use the last $500,000$ observations as the validation set. 

Figures~\ref{fig:7-8} gives the MSEs and proportions of correct classification when the classification probability threshold is
0.5. It is seen that the OSMAC with $\boldsymbol{\pi}^{\mathrm{mMSE}}$ always results in the smallest MSEs. For classifications, the result from the full data is better than the subsampling methods, but the difference is not significant. Among the three subsampling methods, the OSMAC with $\boldsymbol{\pi}^{\mathrm{mVc}}$ has the best performance in classification.

\begin{figure}[htp]
  \centering
  \begin{subfigure}{0.495\textwidth}
    \includegraphics[width=\textwidth]{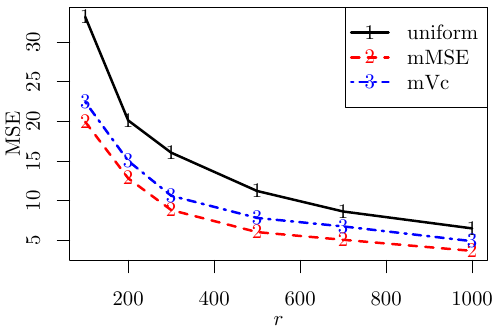}
    \caption{MSEs vs  ${r}$}
  \end{subfigure}
  \begin{subfigure}{0.495\textwidth}
    \includegraphics[width=\textwidth]{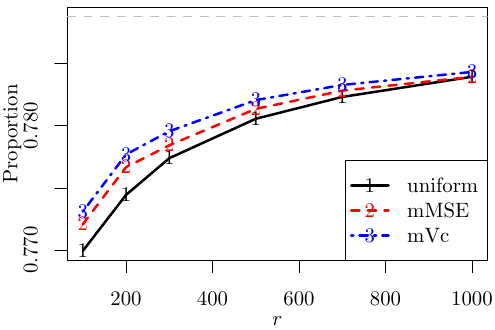}
    \caption{Proportions of correct classifications vs $r$}
  \end{subfigure}  
  \caption{MSEs and proportions of correct classifications for for the SUSY data set with $r_0=200$ and different second step subsample size ${r}$. The gray horizontal dashed line in figure (b) is the result using the full data MLE.}
  \label{fig:7-8}
\end{figure}

To further evaluate the performance of the OSMAC methods as
classifiers, we create receiver operating characteristic (ROC) curves
using classification probability thresholds between 0 and 1, and then
calculate the areas under the ROC curves (AUC). As pointed out
  by an Associate Editor, the theoretical investigation of this paper
  focus on parameter estimation, and classification is not
  theoretically studied. These two goals, although being different,
  are highly connected. Since logistic regression models are commonly
  used for classification, the relevant performance is also important for practical application.
 Table~\ref{tab:4}
presents the results based on 1000 subsamples of size $r_0+r=1000$ from
the full data. For the two step algorithm, $r_0=200$ and ${r}=800$.
All the AUCs are around 0.85, meaning that the classifiers all have
good performance.

For the same data set considered here, the deep learning method (DL)
in \cite{Baldi2014} produced an AUC of 0.88 while the AUCs from the
OSMAC approach with different SSPs are around 0.85. However the DL used the full data set and had to optimize a much more complicated model
(``{\it a five-layer neural network with 300 hidden units in each
  layer}'' \citep{Baldi2014}), while OSMAC just used $r=1000$
observations and the target function to optimize is the log-likelihood
function of a logistic regression model. Due to computational costs,
the optimization in \cite{Baldi2014} included ``{\it combinations of
  the pre-training methods, network architectures, initial learning
  rates, and regularization methods.}'' For the OSMAC, the
optimization was done by using a standard Newton's method
directly. The computations of \cite{Baldi2014} ``{\it were performed
  using machines with 16 Intel Xeon cores, an NVIDIA Tesla C2070
  graphics processor, and 64 GB memory. All neural networks were
  trained using the GPU-accelerated Theano and Pylearn2 software
  libraries}''. Our analysis was just carried out on a normal PC with
an Intel I7 processor and 16GB memory. Clearly, \cite{Baldi2014}'s
method requires special computing resources and coding skills, but
anyone with basic programming ability is able to implement the OSMAC.  Due to the special requirements of \cite{Baldi2014}'s
method, we are not able to replicate their results and thus cannot report the computing time. For our OSMAC with $\boldsymbol{\pi}^{\mathrm{mMSE}}$ and $\boldsymbol{\pi}^{\mathrm{mVc}}$, the average CPU seconds to obtain parameter estimates are 3.400 and 1.079 seconds, respectively. The full data MLE takes an average of 24.060 seconds to run.

\begin{table}[htp]
  \centering
  \caption{Average AUC (as percentage) for the SUSY data set based on
    1000 subsamples. A number in the parentheses is the associated
    standard error (as percentage) of the 1000 AUCs.}
  \begin{tabular}{ccc} \hline
    Method      && AUC \%  (SE)     \\ \hline
    uniform     && 85.06  (0.29) \\
    mMSE        && 85.08  (0.30) \\
    mVc         && 85.17  (0.25) \\
    Full  && 85.75         \\
    \hline
  \end{tabular}
  \label{tab:4}
\end{table}

\section{Discussion}\label{sec:discussion}
In this paper, we proposed the OSMAC approach for logistic regression
in order to overcome the computation bottleneck resulted from the
explosive growth of data volume. Not only were theoretical statistical
asymptotic results derived, but also optimal subsampling methods
were given. Furthermore, we developeded a two-step subsampling
algorithm to approximate optimal subsampling strategies and proved that the resultant estimator is consistent and asymptotically normal with the optimal variance-covariance matrix.  As shown in
our numerical experiments, the OSMAC approach for logistic regression
is a computationally feasible method for super-large samples, and it
yields a good approximation to the results based on full data. There are
important issues in this paper that we will investigate in the
future. 

\begin{enumerate}
\item In our numerical experiments, the formula in \eqref{eq:24}
    underestimates the asymptotic variance-covariance matrix and thus
    does not produce an accurate approximation for rare events
    data. It is unclear whether the technique in \cite{king2001logistic} can be
    applied to develop an improved estimator of the asymptotic
    variance-covariance matrix in the case of rare event. It is an interesting question worth further investigations. 
\item We have chosen to minimize the trace of ${\mathbf{V}}$ or ${\mathbf{V}}_c$ to define
  optimal subsampling algorithms. The idea is from the $A$-optimality
  criterion in the theory of optimal experimental designs
  \citep{kiefer1959}. There are other optimality criteria emphasizing
  different inferential purposes, such as the $C$-optimality and the
  $D$-optimality. How to use these optimality criteria to develop high
  quality algorithms is undoubtedly a topic worthy of future study.
\end{enumerate}

{

}

\pagebreak
\begin{center}
  {\Large Supplementary Material\\ for ``Optimal Subsampling for Large Sample Logistic Regression''}
\end{center}

\setcounter{equation}{0}
\setcounter{section}{0}
\setcounter{subsection}{0}
\renewcommand{\theequation}{S.\arabic{equation}}
\renewcommand{\thesection}{S.\arabic{section}}
\setcounter{figure}{0}
\renewcommand{\thefigure}{S.\arabic{figure}}
\setcounter{assumption}{0}
\renewcommand{\theassumption}{S.\arabic{assumption}}

\section{Proofs}
In this section we prove the theorems in the paper. 
\subsection{Proof of Theorem \ref{thm:1}}
We begin by establishing a lemma that will be used in the proof of Theorems~\ref{thm:1} and
\ref{thm:2}.

\begin{lemma}\label{lem1}
If Assumptions~\ref{asp:1} and \ref{asp:2} hold, then conditionally on ${\mathcal{F}_n}$ in probability,
\begin{align}
  \tilde{{\mathbf{M}}}_X-{\mathbf{M}}_X&=O_{P|{\mathcal{F}_n}}(r^{-1/2}),\label{eq:8}\\ \frac{1}{n}\frac{\partial\ell^*(\hat{{\boldsymbol{\beta}}}_{{{\textnormal{\tiny MLE}}}})}{\partial{\boldsymbol{\beta}}} &=O_{P|{\mathcal{F}_n}}(r^{-1/2}),\label{eq:9}
\end{align}
where
\begin{equation*}
  \tilde{{\mathbf{M}}}_X=
  \frac{1}{n}\frac{\partial^2\ell^*(\hat{{\boldsymbol{\beta}}}_{{{\textnormal{\tiny MLE}}}})}
  {\partial{\boldsymbol{\beta}}\partial{\boldsymbol{\beta}}^T}
  =\frac{1}{nr}\sum_{i=1}^r\frac{w_i^*(\hat{{\boldsymbol{\beta}}}_{{{\textnormal{\tiny MLE}}}}){\mathbf{x}}_i^*({\mathbf{x}}_i^*)^T}{\pi^*_i}.
\end{equation*}
\end{lemma}

\begin{proof}
Direct calculation yields
\begin{equation}\label{eq:5}
    {\mathrm{E}}(\tilde{{\mathbf{M}}}_X|{\mathcal{F}_n})={\mathbf{M}}_X.
\end{equation}
For any component $\tilde{{\mathbf{M}}}_X^{j_1j_2}$ of $\tilde{{\mathbf{M}}}_X$ where $1\le j_1,j_2\le d$,
\begin{align*}
  {\mathrm{Var}}\left(\frac{1}{n}\tilde{{\mathbf{M}}}_X^{j_1j_2}\Big|{\mathcal{F}_n}\right)
  =&\frac{1}{r}\sum_{i=1}^n\pi_i
     \left\{\frac{w_i(\hat{{\boldsymbol{\beta}}}_{{{\textnormal{\tiny MLE}}}})x_{ij_1}x_{ij_2}}{n\pi_i}
     -{\mathbf{M}}_X^{j_1j_2}\right\}^2\\
  =&\frac{1}{rn^2}\sum_{i=1}^n
     \frac{w_i(\hat{{\boldsymbol{\beta}}}_{{{\textnormal{\tiny MLE}}}})^2(x_{ij_1}x_{ij_2}^T)^2}{\pi_i}
     -\frac{1}{r}({\mathbf{M}}_X^{j_1j_2})^2\\
  \le&\frac{1}{16rn^2}\sum_{i=1}^n\frac{\|{\mathbf{x}}_i\|^4}{\pi_i}
       -\frac{1}{r}({\mathbf{M}}_X^{j_1j_2})^2\\
  =&O_P(r^{-1}),
\end{align*}
where the second last inequality holds by the fact that $0<w_i(\hat{{\boldsymbol{\beta}}}_{{{\textnormal{\tiny MLE}}}})\leq1/4$ and the last equality is from Assumption \ref{asp:2}. Using Markov's inequality, this result and \eqref{eq:5}, implies \eqref{eq:8}.

To prove \eqref{eq:9},  direct calculation yields,
\begin{equation}\label{eq:10}
  {\mathrm{E}}\left\{\frac{1}{n}\frac{\partial\ell^*(\hat{{\boldsymbol{\beta}}}_{{{\textnormal{\tiny MLE}}}})}
    {\partial{\boldsymbol{\beta}}}\bigg|{\mathcal{F}_n}\right\}
  =\frac{1}{nr}\frac{\partial\ell^*(\hat{{\boldsymbol{\beta}}}_{{{\textnormal{\tiny MLE}}}})}
      {\partial{\boldsymbol{\beta}}}= 0.
\end{equation}
From Assumption \ref{asp:2},
\begin{equation}\label{eq:11}
{\mathrm{Var}}\left\{\frac{1}{n}\frac{\partial\ell^*(\hat{{\boldsymbol{\beta}}}_{{{\textnormal{\tiny MLE}}}})}
  {\partial{\boldsymbol{\beta}}}\bigg|{\mathcal{F}_n}\right\}
=\frac{1}{n^2r}\sum_{i=1}^n\frac{\{y_i-p_i(\hat{{\boldsymbol{\beta}}}_{{{\textnormal{\tiny MLE}}}})\}^2{\mathbf{x}}_i{\mathbf{x}}_i^T}{\pi_i}
\le\frac{1}{n^2r}\sum_{i=1}^n\frac{{\mathbf{x}}_i{\mathbf{x}}_i^T}{\pi_i}
=O_P(r^{-1}).
\end{equation}
From \eqref{eq:10}, \eqref{eq:11} and Markov's inequality, \eqref{eq:9} follows.
\end{proof}

Now we prove Theorem \ref{thm:1}. 
Note that $t_i({\boldsymbol{\beta}})=y_i\log p_i({\boldsymbol{\beta}})+(1-y_i)\log\{1-
p_i({\boldsymbol{\beta}})\}$, $t_i^*({\boldsymbol{\beta}})=y_i^*\log p_i^*({\boldsymbol{\beta}})+(1-y_i^*)\log\{1-
p_i^*({\boldsymbol{\beta}})\}$,
\begin{equation*}
\ell^*({\boldsymbol{\beta}})=\frac{1}{r}\sum_{i=1}^r \frac{t_i^*({\boldsymbol{\beta}})}{\pi_i^*}, \ \ \text{and} \ \ \ell({\boldsymbol{\beta}})=\sum_{i=1}^n t_i({\boldsymbol{\beta}}).
\end{equation*}
By direct calculation under the conditional distribution of subsample given ${\mathcal{F}_n}$,
\begin{equation}\label{eq:2}
{\mathrm{E}}\left\{\frac{\ell^*({\boldsymbol{\beta}})}{n}-\frac{\ell({\boldsymbol{\beta}})}{n}\bigg|{\mathcal{F}_n}\right\}^2
  =\frac{1}{r}\left[\frac{1}{n^2}\sum_{i=1}^n
  \frac{t_i^2({\boldsymbol{\beta}})}{\pi_i}-\left(\frac{1}{n}\sum_{i=1}^n t_i({\boldsymbol{\beta}})\right)^2\right].
\end{equation}
Note that $|t_i({\boldsymbol{\beta}})| \le\log4+2\|{\mathbf{x}}_i\|\|{\boldsymbol{\beta}}\|$. Therefore, from Assumption~\ref{asp:1},
\begin{align}\label{moment-condition}
  \frac{1}{n^2}\sum_{i=1}^n
  \frac{t_i^2({\boldsymbol{\beta}})}{\pi_i}-\left(\frac{1}{n}\sum_{i=1}^n t_i({\boldsymbol{\beta}})\right)^2
  \le&\frac{1}{n^2}\sum_{i=1}^n\frac{t_i^2({\boldsymbol{\beta}})}{\pi_i}
       +\left(\frac{1}{n}\sum_{i=1}^n|t_i({\boldsymbol{\beta}})|\right)^2=O_P(1).
\end{align}
Therefore combing \eqref{eq:2} and \eqref{moment-condition}, $n^{-1}\ell^*({\boldsymbol{\beta}})-n^{-1}\ell({\boldsymbol{\beta}})\rightarrow0$ in conditional probability given ${\mathcal{F}_n}$. Note that the parameter space is compact and $\hat{{\boldsymbol{\beta}}}_{{{\textnormal{\tiny MLE}}}}$ is the unique global maximum of the continuous convex function $\ell({\boldsymbol{\beta}})$. Thus, from Theorem 5.9 and its remark of \cite{Vaart:98}, conditionally on ${\mathcal{F}_n}$ in probability,
\begin{equation}\label{eq:12}
  \|\tilde{{\boldsymbol{\beta}}}-\hat{{\boldsymbol{\beta}}}_{{{\textnormal{\tiny MLE}}}}\|=o_{P|{\mathcal{F}_n}}(1)
\end{equation}

The consistency proved above ensures that $\tilde{{\boldsymbol{\beta}}}$ is close to $\hat{{\boldsymbol{\beta}}}_{{{\textnormal{\tiny MLE}}}}$ as long as $r$ is not small. Using Taylor's theorem \citep[c.f. Chapter 4 of][]{Ferguson1996},

\begin{align}
0=\frac{\dot\ell^*_j(\tilde{{\boldsymbol{\beta}}})}{n}
=&\frac{\dot\ell^*_j(\hat{{\boldsymbol{\beta}}}_{{{\textnormal{\tiny MLE}}}})}{n}
   +\frac{1}{n} \frac{\partial\dot\ell^*_j(\hat{{\boldsymbol{\beta}}}_{{{\textnormal{\tiny MLE}}}})}{\partial {\boldsymbol{\beta}}^T}(\tilde{{\boldsymbol{\beta}}}-\hat{{\boldsymbol{\beta}}}_{{{\textnormal{\tiny MLE}}}})
   +\frac{1}{n} R_j
   \label{taylor-expansion}
\end{align}
where $\dot\ell^*_j({{\boldsymbol{\beta}}})$ is the partial derivative of $\ell^*({{\boldsymbol{\beta}}})$ with respect to $\beta_j$, and
\begin{equation*}
  R_j=(\tilde{{\boldsymbol{\beta}}}-\hat{{\boldsymbol{\beta}}}_{{{\textnormal{\tiny MLE}}}})^T
   \int_0^1\int_0^1\frac{\partial^2\dot\ell^*_j\{\hat{{\boldsymbol{\beta}}}_{{{\textnormal{\tiny MLE}}}}
   +uv(\tilde{{\boldsymbol{\beta}}}-\hat{{\boldsymbol{\beta}}}_{{{\textnormal{\tiny MLE}}}})\}}{\partial {\boldsymbol{\beta}}\partial
   {\boldsymbol{\beta}}^T}v{\mathrm{d}} u{\mathrm{d}} v\
(\tilde{{\boldsymbol{\beta}}}-\hat{{\boldsymbol{\beta}}}_{{{\textnormal{\tiny MLE}}}}).
\end{equation*}
Note that
\begin{align*}
  \left\|\frac{\partial^2\dot\ell^*_j({\boldsymbol{\beta}})}
  {\partial{\boldsymbol{\beta}}\partial{\boldsymbol{\beta}}^T}\right\|=
  &\frac{1}{r}\left\|\sum_{i=1}^r \frac{p^*_i({\boldsymbol{\beta}})\{1-p^*_i({\boldsymbol{\beta}})\}
    \{1-2p^*_i({\boldsymbol{\beta}})\}}{\pi^*_i}x^*_{ij}{\mathbf{x}}^*_i{{\mathbf{x}}^*_i}^T\right\|
    \le\frac{1}{r}\sum_{i=1}^r\frac{\|{\mathbf{x}}^*_i\|^3}{\pi^*_i}
\end{align*}
for all ${\boldsymbol{\beta}}$. Thus
\begin{align}\label{eq:3}
  \left\|\int_0^1\int_0^1\frac{\partial^2\dot\ell^*_j\{\hat{{\boldsymbol{\beta}}}_{{{\textnormal{\tiny MLE}}}}
   +uv(\tilde{{\boldsymbol{\beta}}}-\hat{{\boldsymbol{\beta}}}_{{{\textnormal{\tiny MLE}}}})\}}{\partial {\boldsymbol{\beta}}\partial
  {\boldsymbol{\beta}}^T}v{\mathrm{d}} u{\mathrm{d}} v\ \right\|
  \le\frac{1}{2r}\sum_{i=1}^r\frac{\|{\mathbf{x}}^*_i\|^3}{\pi^*_i}=O_{P|{\mathcal{F}_n}}(n),
\end{align}
where the last equality is from the fact that
\begin{align}\label{eq:53}
  P\left(\frac{1}{nr}\sum_{i=1}^r\frac{\|{\mathbf{x}}^*_i\|^3}{\pi^*_i}\ge\tau
  \Bigg|{\mathcal{F}_n}\right) \le\frac{1}{nr\tau}\sum_{i=1}^r
  {\mathrm{E}}\left(\frac{\|{\mathbf{x}}^*_i\|^3}{\pi^*_i}\Bigg|{\mathcal{F}_n}\right)
  =\frac{1}{n\tau}\sum_{i=1}^n\|{\mathbf{x}}_i\|^3\rightarrow0,
\end{align}
in probability as $\tau\rightarrow\infty$ by Assumption~\ref{asp:2}.
From \eqref{taylor-expansion} and \eqref{eq:3},
\begin{equation}\label{eq:13}
  \tilde{{\boldsymbol{\beta}}}-\hat{{\boldsymbol{\beta}}}_{{{\textnormal{\tiny MLE}}}}=
  -\tilde{{\mathbf{M}}}_X^{-1}\left\{\frac{\dot \ell^*(\hat{{\boldsymbol{\beta}}}_{{{\textnormal{\tiny MLE}}}})}{n}
    +O_{P|{\mathcal{F}_n}}(\|\tilde{{\boldsymbol{\beta}}}-\hat{{\boldsymbol{\beta}}}_{{{\textnormal{\tiny MLE}}}}\|^2)\right\}.
\end{equation}
From \eqref{eq:8} of Lemma~\ref{lem1}, $\tilde{{\mathbf{M}}}_X^{-1} =O_{P|{\mathcal{F}_n}}(1)$. Combining this with \eqref{eq:9}, \eqref{eq:12} and \eqref{eq:13}
\begin{equation*}
  \tilde{{\boldsymbol{\beta}}}-\hat{{\boldsymbol{\beta}}}_{{{\textnormal{\tiny MLE}}}}
  =O_{P|{\mathcal{F}_n}}(r^{-1/2})+o_{P|{\mathcal{F}_n}}(\|\tilde{{\boldsymbol{\beta}}}-\hat{{\boldsymbol{\beta}}}_{{{\textnormal{\tiny MLE}}}}\|),
\end{equation*}
which implies that
\begin{equation}\label{eq:6}
  \tilde{{\boldsymbol{\beta}}}-\hat{{\boldsymbol{\beta}}}_{{{\textnormal{\tiny MLE}}}}=O_{P|{\mathcal{F}_n}}(r^{-1/2}).
\end{equation}

\subsection{Proof of Theorem \ref{thm:2}}
  Note that
\begin{equation}\label{eq:14}
  \frac{\dot \ell^*(\hat{{\boldsymbol{\beta}}}_{{{\textnormal{\tiny MLE}}}})}{n}
  =\frac{1}{r}\sum_{i=1}^r\frac{\{y^*_i-p^*_i(\hat{{\boldsymbol{\beta}}}_{{{\textnormal{\tiny MLE}}}})\}{\mathbf{x}}^*_i}{n\pi^*_i}
  \equiv\frac{1}{r}\sum_{i=1}^r{\boldsymbol{\eta}}_i
\end{equation}
Given ${\mathcal{F}_n}$, ${\boldsymbol{\eta}}_1, ..., {\boldsymbol{\eta}}_r$ are i.i.d, with mean $\mathbf{0}$ and variance,
\begin{align}
  &{\mathrm{Var}}({\boldsymbol{\eta}}_i|{\mathcal{F}_n})=r{\mathbf{V}}_c=
 \frac{1}{n^2}\sum_{i=1}^n\frac{\{y_i-p_i(\hat{{\boldsymbol{\beta}}}_{{{\textnormal{\tiny MLE}}}})\}^2
 {\mathbf{x}}_i{\mathbf{x}}_i^T}{\pi_i} =O_P(1).\label{condition-moments}
\end{align}
Meanwhile, for every $\varepsilon>0$ and some $\delta>0$,
\begin{align*}
  & \sum_{i=1}^r {\mathrm{E}}\{\|r^{-1/2}{\boldsymbol{\eta}}_i\|^2
    I(\|{\boldsymbol{\eta}}_i\|>r^{1/2}\varepsilon)|{\mathcal{F}_n}\}\\
  &\le\frac{1}{r^{1+\delta/2}\varepsilon^{\delta}}
    \sum_{i=1}^r{\mathrm{E}}\{\|{\boldsymbol{\eta}}_i\|^{2+\delta}
    I(\|{\boldsymbol{\eta}}_i\|>r^{1/2}\varepsilon)|{\mathcal{F}_n}\}\\
  &\le\frac{1}{r^{1+\delta/2}\varepsilon^{\delta}}
    \sum_{i=1}^r {\mathrm{E}}(\|{\boldsymbol{\eta}}_i\|^{2+\delta}|{\mathcal{F}_n})\\
  &=\frac{1}{r^{\delta/2}}\frac{1}{n^{2+\delta}}\frac{1}{\varepsilon^{\delta}}
    \sum_{i=1}^n\frac{\{y_i-p_i(\hat{{\boldsymbol{\beta}}}_{{{\textnormal{\tiny MLE}}}})\}^{2+\delta}
    \|{\mathbf{x}}_i\|^{2+\delta}}{\pi_i^{1+\delta}}\\
  &\le\frac{1}{r^{\delta/2}}\frac{1}{n^{2+\delta}}\frac{1}{\varepsilon^{\delta}}
    \sum_{i=1}^n\frac{\|{\mathbf{x}}_i\|^{2+\delta}}{\pi_i^{1+\delta}}=o_P(1)
\end{align*}
where the last equality is from Assumption \ref{asp:3}. This and \eqref{condition-moments} show that the Lindeberg-Feller conditions are satisfied in probability.
From \eqref{eq:14} and \eqref{condition-moments}, by the Lindeberg-Feller central limit theorem \citep[Proposition 2.27 of][]{Vaart:98}, conditionally on ${\mathcal{F}_n}$,
\begin{equation*}
  \frac{1}{n}{\mathbf{V}}_c^{-1/2}\dot\ell^*(\hat{{\boldsymbol{\beta}}}_{{{\textnormal{\tiny MLE}}}})=
  \frac{1}{r^{1/2}}\{{\mathrm{Var}}({\boldsymbol{\eta}}_i|{\mathcal{F}_n})\}^{-1/2}\sum_{i=1}^r{\boldsymbol{\eta}}_i
  \rightarrow N(0,\mathbf{I}),
\end{equation*}
in distribution.
From Lemma~\ref{lem1}, \eqref{eq:13} and~\eqref{eq:6},
\begin{equation}\label{eq:15}
  \tilde{{\boldsymbol{\beta}}}-\hat{{\boldsymbol{\beta}}}_{{{\textnormal{\tiny MLE}}}}=
  -\frac{1}{n}\tilde{{\mathbf{M}}}_X^{-1}\dot\ell^*(\hat{{\boldsymbol{\beta}}}_{{{\textnormal{\tiny MLE}}}})+O_{P|{\mathcal{F}_n}}(r^{-1})
\end{equation}
From \eqref{eq:8} of Lemma~\ref{lem1},
\begin{align}\label{eq:7}
  \tilde{{\mathbf{M}}}_X^{-1}-{\mathbf{M}}_X^{-1}
  &=-{\mathbf{M}}_X^{-1}(\tilde{{\mathbf{M}}}_X-{\mathbf{M}}_X)\tilde{{\mathbf{M}}}_X^{-1}
  =O_{P|{\mathcal{F}_n}}(r^{-1/2}).
\end{align}
Based on Assumption~\ref{asp:1} and \eqref{condition-moments}, it is verified that,
\begin{equation}\label{eq:16}
  {\mathbf{V}}=\boldsymbol{{\mathbf{M}}}_X^{-1}{\mathbf{V}}_c\boldsymbol{{\mathbf{M}}}_X^{-1}
  =\frac{1}{r}\boldsymbol{{\mathbf{M}}}_X^{-1}\left(r{\mathbf{V}}_c\right)
  \boldsymbol{{\mathbf{M}}}_X^{-1}=O_{P}(r^{-1}).
\end{equation}
Thus, \eqref{eq:15}, \eqref{eq:7} and \eqref{eq:16} yield,
\begin{align*}
  {\mathbf{V}}^{-1/2}(\tilde{{\boldsymbol{\beta}}}-\hat{{\boldsymbol{\beta}}}_{{{\textnormal{\tiny MLE}}}})
  &=-{\mathbf{V}}^{-1/2}n^{-1}\tilde{{\mathbf{M}}}_X^{-1}\dot\ell^*(\hat{{\boldsymbol{\beta}}}_{{{\textnormal{\tiny MLE}}}})
    +O_{P|{\mathcal{F}_n}}(r^{-1/2})\\
  &=-{\mathbf{V}}^{-1/2}{\mathbf{M}}_X^{-1}n^{-1}\dot\ell^*(\hat{{\boldsymbol{\beta}}}_{{{\textnormal{\tiny MLE}}}})
    -{\mathbf{V}}^{-1/2}(\tilde{{\mathbf{M}}}_X^{-1}-{\mathbf{M}}_X^{-1})n^{-1}\dot\ell^*(\hat{{\boldsymbol{\beta}}}_{{{\textnormal{\tiny MLE}}}})
    +O_{P|{\mathcal{F}_n}}(r^{-1/2})\\
  &=-{\mathbf{V}}^{-1/2}{\mathbf{M}}_X^{-1}{\mathbf{V}}_c^{1/2}{\mathbf{V}}_c^{-1/2}n^{-1}\dot\ell^*(\hat{{\boldsymbol{\beta}}}_{{{\textnormal{\tiny MLE}}}})
    +O_{P|{\mathcal{F}_n}}(r^{-1/2}).
\end{align*}
The result in \eqref{normal} of Theorem~\ref{thm:1} follows from Slutsky's Theorem\citep[Theorem 6 of][]{Ferguson1996} and the fact that
\begin{equation*}
  {\mathbf{V}}^{-1/2}{\mathbf{M}}_X^{-1}{\mathbf{V}}_c^{1/2}({\mathbf{V}}^{-1/2}{\mathbf{M}}_X^{-1}{\mathbf{V}}_c^{1/2})^T
  ={\mathbf{V}}^{-1/2}{\mathbf{M}}_X^{-1}{\mathbf{V}}_c^{1/2}{\mathbf{V}}_c^{1/2}{\mathbf{M}}_X^{-1}{\mathbf{V}}^{-1/2}=\mathbf{I}.
\end{equation*}


\subsection{Proof of Theorems~\ref{thm:3} and \ref{thm:4}}
For Theorem~\ref{thm:3},
\begin{align*}
  \mathrm{tr}({\mathbf{V}})=\mathrm{tr}({\mathbf{M}}_X^{-1}{\mathbf{V}}_c{\mathbf{M}}_X^{-1})
  &=\frac{1}{r}\sum_{i=1}^n\mathrm{tr}\left[\frac{1}{\pi_i}
    \{y_i-p_i(\hat{{\boldsymbol{\beta}}}_{{{\textnormal{\tiny MLE}}}})\}^2{\mathbf{M}}_X^{-1}{\mathbf{x}}_i{\mathbf{x}}_i^T{\mathbf{M}}_X^{-1}\right]\\
  &=\frac{1}{r}\sum_{i=1}^n\left[\frac{1}{\pi_i}\{y_i-p_i(\hat{{\boldsymbol{\beta}}}_{{{\textnormal{\tiny MLE}}}})\}^2
    \|{\mathbf{M}}_X^{-1}{\mathbf{x}}_i\|^2\right]\\
  &=\frac{1}{r}\sum_{i=1}^n\pi_i\sum_{i=1}^n\left[\pi_i^{-1}\{y_i-p_i(\hat{{\boldsymbol{\beta}}}_{{{\textnormal{\tiny MLE}}}})\}^2
    \|{\mathbf{M}}_X^{-1}{\mathbf{x}}_i\|^2\right]\\
  &\ge\frac{1}{r}\left[\sum_{i=1}^n|y_i-p_i(\hat{{\boldsymbol{\beta}}}_{{{\textnormal{\tiny MLE}}}})
    |\|{\mathbf{M}}_X^{-1}{\mathbf{x}}_i\|\right]^2,
\end{align*}
where the last step is from the Cauchy-Schwarz inequality and the equality in it holds if and only if when $\pi_i\propto|y_i-p_i(\hat{{\boldsymbol{\beta}}}_{{{\textnormal{\tiny MLE}}}})|\|{\mathbf{M}}_X^{-1}{\mathbf{x}}_i\|$.

The proof of Theorem~\ref{thm:4} is similar to the proof of Theorem~\ref{thm:3} and thus is omit it to save space.

{
  \subsection{Proof of Theorems~\ref{thm:5}}
  \label{sec:proof-theor-refthm:5}
Since $r_0r^{-1/2}\rightarrow0$, the contribution of the first step subsample to the likelihood function is a small term with an order $o_{P|{\mathcal{F}_n}}(r^{-1/2})$ relative the likelihood function. Thus, we can focus on the second step subsample only. Denote
\begin{equation*}
\ell^*_{\tilde{{\boldsymbol{\beta}}}_0}({\boldsymbol{\beta}})=\frac{1}{{r}}\sum_{i=1}^{{r}} \frac{t_i^*({\boldsymbol{\beta}})}{\pi_i^*(\tilde{{\boldsymbol{\beta}}}_0)},
\end{equation*}
where $\pi_i^*(\tilde{{\boldsymbol{\beta}}}_0)$ has the same expression as $\pi_i^{\mathrm{mVc}}$ except that $\hat{{\boldsymbol{\beta}}}_{{{\textnormal{\tiny MLE}}}}$ is replaced by $\tilde{{\boldsymbol{\beta}}}_0$. We first establish two lemmas that will be used in the proof of Theorems~\ref{thm:5} and \ref{thm:6}.

\begin{lemma}\label{lem2}
  Let the compact parameter space be $\Theta$ and $\lambda=\sup_{{\boldsymbol{\beta}}\in\Theta}\|{\boldsymbol{\beta}}\|$. Under Assumption~\ref{asp:4}, for $k_1\ge k_2\ge0$,
\begin{align}
  \frac{1}{n^2}\sum_{i=1}^n\frac{\|{\mathbf{x}}_i\|^{k_1}}{\pi_i^{k_2}(\tilde{{\boldsymbol{\beta}}}_0)}
  \le&\frac{3^{k_2}}{n}\sum_{i=1}^n\|{\mathbf{x}}_i\|^{k_1-k_2}
    e^{\lambda k_2\|{\mathbf{x}}_i\|}
     \frac{1}{n}\sum_{i=1}^n\|{\mathbf{x}}_i\|^{k_2}=O_P(1).\label{eq:34}
\end{align}
\end{lemma}
\begin{proof}
From the expression of $\pi_i(\tilde{{\boldsymbol{\beta}}}_0)$,
\begin{align}
  \frac{1}{n^2}\sum_{i=1}^n\frac{\|{\mathbf{x}}_i\|^{k_1}}{\pi_i^{k_2}(\tilde{{\boldsymbol{\beta}}}_0)}
  =&\frac{1}{n}\sum_{i=1}^n\frac{\|{\mathbf{x}}_i\|^{k_1-k_2}}
     {|y_i-p_i(\tilde{{\boldsymbol{\beta}}}_0)|^{k_2}}\;
     \frac{1}{n}\sum_{j=1}^n|y_j-p_j(\tilde{{\boldsymbol{\beta}}}_0)|^{k_2}\|{\mathbf{x}}_j\|^{k_2}.
     \label{eq:35}
\end{align}
For the first term on the right hand side of \eqref{eq:35},
  \begin{align}
  \frac{1}{n}\sum_{i=1}^n\frac{\|{\mathbf{x}}_i\|^{k_1-k_2}}{|y_i-p_i(\tilde{{\boldsymbol{\beta}}}_0)|^{k_2}}
  &\le\frac{1}{n}\sum_{i=1}^n\|{\mathbf{x}}_i\|^{k_1-k_2}
  (1+e^{{\mathbf{x}}_i^T\tilde{{\boldsymbol{\beta}}}_0}+e^{-{\mathbf{x}}_i^T\tilde{{\boldsymbol{\beta}}}_0})^{k_2}\notag\\
  &\le\frac{1}{n}\sum_{i=1}^n\|{\mathbf{x}}_i\|^{k_1-k_2}
    (1+2e^{\|{\mathbf{x}}_i\|\|\tilde{{\boldsymbol{\beta}}}_0\|})^{k_2}\notag\\
  &\le\frac{3^{k_2}}{n}\sum_{i=1}^n\|{\mathbf{x}}_i\|^{k_1-k_2}
    e^{\lambda k_2\|{\mathbf{x}}_i\|}.\label{eq:36}
\end{align}
Note that
\begin{align}
  {\mathrm{E}}\{\|{\mathbf{x}}_i\|^{k_1-k_2}e^{\lambda k_2\|{\mathbf{x}}_i\|}\}
  \le\{{\mathrm{E}}(\|{\mathbf{x}}_i\|^{2(k_1-k_2)})
  {\mathrm{E}}(e^{2\lambda k_2\|{\mathbf{x}}_i\|})\}^{1/2}\le\infty.\label{eq:37}
\end{align}
Combining \eqref{eq:35}, \eqref{eq:36} and \eqref{eq:37}, and using the Law of Large Numbers, \eqref{eq:34} follows.
\end{proof}

The following lemma is similar to Lemma~\ref{lem1}.
\begin{lemma}\label{lem3}
If Assumption~\ref{asp:4} holds, then conditionally on ${\mathcal{F}_n}$ in probability,
\begin{align}
  \tilde{{\mathbf{M}}}_X^{\tilde{{\boldsymbol{\beta}}}_0}-{\mathbf{M}}_X&=O_{P|{\mathcal{F}_n}}({r}^{-1/2}),\label{eq:4}\\ \frac{1}{n}\frac{\partial\ell^*_{\tilde{{\boldsymbol{\beta}}}_0}(\hat{{\boldsymbol{\beta}}}_{{{\textnormal{\tiny MLE}}}})}{\partial{\boldsymbol{\beta}}} &=O_{P|{\mathcal{F}_n}}({r}^{-1/2}),\label{eq:28}
\end{align}
where
\begin{equation*}
  \tilde{{\mathbf{M}}}_X^{\tilde{{\boldsymbol{\beta}}}_0}=
  \frac{1}{n}\frac{\partial^2\ell^*_{\tilde{{\boldsymbol{\beta}}}_0}(\hat{{\boldsymbol{\beta}}}_{{{\textnormal{\tiny MLE}}}})}
  {\partial{\boldsymbol{\beta}}\partial{\boldsymbol{\beta}}^T}
  =\frac{1}{n{r}}\sum_{i=1}^{{r}}\frac{w_i^*(\hat{{\boldsymbol{\beta}}}_{{{\textnormal{\tiny MLE}}}})
    {\mathbf{x}}_i^*({\mathbf{x}}_i^*)^T}{\pi^*_i(\tilde{{\boldsymbol{\beta}}}_0)}.
\end{equation*}
\end{lemma}

\begin{proof}
Direct calculation yields,
\begin{equation}
  {\mathrm{E}}(\tilde{{\mathbf{M}}}_X|{\mathcal{F}_n})
  ={\mathrm{E}}_{\tilde{{\boldsymbol{\beta}}}_0}\{{\mathrm{E}}(\tilde{{\mathbf{M}}}_X|{\mathcal{F}_n},\tilde{{\boldsymbol{\beta}}}_0)\}
  ={\mathrm{E}}_{\tilde{{\boldsymbol{\beta}}}_0}({\mathbf{M}}_X|{\mathcal{F}_n})
  ={\mathbf{M}}_X,\label{eq:29}
\end{equation}
where ${\mathrm{E}}_{\tilde{{\boldsymbol{\beta}}}_0}$ means the expectation is taken with respect to the distribution of $\tilde{{\boldsymbol{\beta}}}_0$ given ${\mathcal{F}_n}$.
For any component $\tilde{{\mathbf{M}}}_X^{j_1j_2}(\tilde{{\boldsymbol{\beta}}}_0)$ of $\tilde{{\mathbf{M}}}_X^{\tilde{{\boldsymbol{\beta}}}_0}$ where $1\le j_1,j_2\le d$,
\begin{align}
  {\mathrm{Var}}&\left(\frac{1}{n}\tilde{{\mathbf{M}}}_X^{j_1j_2}\Big|
    {\mathcal{F}_n},\tilde{{\boldsymbol{\beta}}}_0\right)\notag\\
  =&\frac{1}{{r}n^2}\sum_{i=1}^n
     \frac{w_i(\hat{{\boldsymbol{\beta}}}_{{{\textnormal{\tiny MLE}}}})^2(x_{ij_1}x_{ij_2}^T)^2}
     {\pi_i(\tilde{{\boldsymbol{\beta}}}_0)}-\frac{1}{{r}}({\mathbf{M}}_X^{j_1j_2})^2\notag\\
  &\le\frac{1}{16{r}n^2}\sum_{i=1}^n\frac{\|{\mathbf{x}}_i\|^4}{\pi_i(\tilde{{\boldsymbol{\beta}}}_0)}
       -\frac{1}{{r}}({\mathbf{M}}_X^{j_1j_2})^2
     \label{eq:30}
\end{align}
From Lemma~\ref{lem2}, and \eqref{eq:30},
\begin{align}
  {\mathrm{Var}}\left(\frac{1}{n}\tilde{{\mathbf{M}}}_X^{j_1j_2}\Big|{\mathcal{F}_n}\right)
  &={\mathrm{E}}_{\tilde{{\boldsymbol{\beta}}}_0}\left\{{\mathrm{Var}}\left(\frac{1}{n}\tilde{{\mathbf{M}}}_X^{j_1j_2}\Big|
    {\mathcal{F}_n},\tilde{{\boldsymbol{\beta}}}_0\right)\right\}\notag\\
  &\le\frac{3}{16{r}}\frac{1}{n}\sum_{j=1}^n\|{\mathbf{x}}_j\|
    \frac{1}{n}\sum_{i=1}^n\|{\mathbf{x}}_i\|^3e^{\lambda\|{\mathbf{x}}_i\|}=O_P({r}^{-1}),
    \label{eq:33}
\end{align}
Using Markov's inequality, \eqref{eq:4} follows from \eqref{eq:29} and \eqref{eq:33}.

Analogously, we obtain that
\begin{equation}
  {\mathrm{E}}\left\{\frac{1}{n}\frac{\partial\ell^*_{\tilde{{\boldsymbol{\beta}}}_0}
      (\hat{{\boldsymbol{\beta}}}_{{{\textnormal{\tiny MLE}}}})}{\partial{\boldsymbol{\beta}}}\bigg|{\mathcal{F}_n}\right\}=0,
  \label{eq:38}
\end{equation}
and
\begin{equation}
{\mathrm{Var}}\left\{\frac{1}{n}\frac{\partial\ell^*(\hat{{\boldsymbol{\beta}}}_{{{\textnormal{\tiny MLE}}}})}
  {\partial{\boldsymbol{\beta}}}\bigg|{\mathcal{F}_n}\right\}
=O_P({r}^{-1}).\label{eq:39}
\end{equation}
From \eqref{eq:38}, \eqref{eq:39} and Markov's inequality,
\eqref{eq:28} follows.
\end{proof}

Now we prove Theorem~\ref{thm:5}.
By direct calculation,
\begin{align}
  &{\mathrm{E}}\left\{\frac{\ell^*_{\tilde{{\boldsymbol{\beta}}}_0}({\boldsymbol{\beta}})}{n}
    -\frac{\ell({\boldsymbol{\beta}})}{n}\bigg|
    {\mathcal{F}_n},\tilde{{\boldsymbol{\beta}}}_0\right\}^2\notag\\
  &=\frac{1}{{r}}\left[\frac{1}{n^2}\sum_{i=1}^n
  \frac{t_i^2({\boldsymbol{\beta}})}{\pi_i(\tilde{{\boldsymbol{\beta}}}_0)}-\left(\frac{1}{n}\sum_{i=1}^n t_i({\boldsymbol{\beta}})\right)^2\right]\notag\\
  &\le\frac{1}{{r}}\left[\frac{1}{n^2}\sum_{i=1}^n
    \frac{(\log4+2\|{\mathbf{x}}_i\|\|{\boldsymbol{\beta}}\|)^2}
    {\pi_i(\tilde{{\boldsymbol{\beta}}}_0)}-\left(\frac{1}{n}\sum_{i=1}^n t_i({\boldsymbol{\beta}})\right)^2\right].\label{eq:40}
\end{align}
Therefore, from Lemma~\ref{lem2} and \eqref{eq:40},
\begin{align}
  &{\mathrm{E}}\left\{\frac{\ell^*_{\tilde{{\boldsymbol{\beta}}}_0}({\boldsymbol{\beta}})}{n}-\frac{\ell({\boldsymbol{\beta}})}{n}\bigg|
    {\mathcal{F}_n}\right\}^2=O_P({r}^{-1}).\label{eq:41}
\end{align}
Therefore combing \eqref{eq:41} and the fact that ${\mathrm{E}}\{\ell^*_{\tilde{{\boldsymbol{\beta}}}_0}({\boldsymbol{\beta}})|{\mathcal{F}_n}\}=\ell({\boldsymbol{\beta}})$, we have  $n^{-1}\ell^*_{\tilde{{\boldsymbol{\beta}}}_0}({\boldsymbol{\beta}})-n^{-1}\ell({\boldsymbol{\beta}})\rightarrow0$ in conditional probability given ${\mathcal{F}_n}$.
Thus, conditionally on ${\mathcal{F}_n}$,
\begin{equation}\label{eq:42}
  \|{\breve{{\boldsymbol{\beta}}}}-\hat{{\boldsymbol{\beta}}}_{{{\textnormal{\tiny MLE}}}}\|=o_{P|{\mathcal{F}_n}}(1)
\end{equation}
The consistency proved above ensures that ${\breve{{\boldsymbol{\beta}}}}$ is close to $\hat{{\boldsymbol{\beta}}}_{{{\textnormal{\tiny MLE}}}}$ as long as ${r}$ is large enough. Using Taylor's theorem \citep[c.f. Chapter 4 of][]{Ferguson1996},
\begin{align}
0=\frac{\dot\ell^*_{\tilde{{\boldsymbol{\beta}}}_0,j}({\breve{{\boldsymbol{\beta}}}})}{n}
=&\frac{\dot\ell^*_{\tilde{{\boldsymbol{\beta}}}_0,j}(\hat{{\boldsymbol{\beta}}}_{{{\textnormal{\tiny MLE}}}})}{n} +\frac{1}{n}
   \frac{\partial\dot\ell^*_{\tilde{{\boldsymbol{\beta}}}_0,j}(\hat{{\boldsymbol{\beta}}}_{{{\textnormal{\tiny MLE}}}})}
   {\partial {\boldsymbol{\beta}}^T}({\breve{{\boldsymbol{\beta}}}}-\hat{{\boldsymbol{\beta}}}_{{{\textnormal{\tiny MLE}}}})
   +\frac{1}{n} R_{\tilde{{\boldsymbol{\beta}}}_0,j}\label{eq:43}
\end{align}
where
$$R_{\tilde{{\boldsymbol{\beta}}}_0,j}=({\breve{{\boldsymbol{\beta}}}}-\hat{{\boldsymbol{\beta}}}_{{{\textnormal{\tiny MLE}}}})^T
\int_0^1\int_0^1\frac{\partial^2\dot\ell^*_{\tilde{{\boldsymbol{\beta}}}_0,j}
  \{\hat{{\boldsymbol{\beta}}}_{{{\textnormal{\tiny MLE}}}}
  +uv(\breve{{\boldsymbol{\beta}}}-\hat{{\boldsymbol{\beta}}}_{{{\textnormal{\tiny MLE}}}})\}}
{\partial {\boldsymbol{\beta}}\partial{\boldsymbol{\beta}}^T}v{\mathrm{d}} u{\mathrm{d}} v\
({\breve{{\boldsymbol{\beta}}}}-\hat{{\boldsymbol{\beta}}}_{{{\textnormal{\tiny MLE}}}}).$$
Note that
\begin{align*}
  \left\|\frac{\partial^2\dot\ell^*_{\tilde{{\boldsymbol{\beta}}}_0,j}({\boldsymbol{\beta}})}
  {\partial{\boldsymbol{\beta}}\partial{\boldsymbol{\beta}}^T}\right\|
  \le\frac{1}{{r}}\sum_{i=1}^{{r}}\frac{\|{\mathbf{x}}^*_i\|^3}
  {\pi^*_i(\tilde{{\boldsymbol{\beta}}}_0)}
\end{align*}
for all ${\boldsymbol{\beta}}$. Thus
\begin{align}\label{eq:44}
  \left\|\int_0^1\int_0^1
  \frac{\partial^2\dot\ell^*_{\tilde{{\boldsymbol{\beta}}}_0,j}\{\hat{{\boldsymbol{\beta}}}_{{{\textnormal{\tiny MLE}}}}
  +uv({\breve{{\boldsymbol{\beta}}}}-\hat{{\boldsymbol{\beta}}}_{{{\textnormal{\tiny MLE}}}})\}}
  {\partial {\boldsymbol{\beta}}\partial{\boldsymbol{\beta}}^T}v{\mathrm{d}} u{\mathrm{d}} v\ \right\|
  \le\frac{1}{2r}\sum_{i=1}^{{r}}\frac{\|{\mathbf{x}}^*_i\|^3}
  {\pi^*_i(\tilde{{\boldsymbol{\beta}}}_0)}=O_{P|{\mathcal{F}_n}}(n),
\end{align}
where the last equality is from the fact that
\begin{align}\label{eq:54}
  P\left(\frac{1}{nr}\sum_{i=1}^{{r}}
  \frac{\|{\mathbf{x}}^*_i\|^3}{\pi^*_i(\tilde{{\boldsymbol{\beta}}}_0)}\ge\tau\Bigg|{\mathcal{F}_n}\right)
  \le\frac{1}{nr\tau}\sum_{i=1}^{{r}}
  {\mathrm{E}}\left(\frac{\|{\mathbf{x}}^*_i\|^3}{\pi^*_i(\tilde{{\boldsymbol{\beta}}}_0)}\Bigg|{\mathcal{F}_n}\right)
  =\frac{1}{n\tau}\sum_{i=1}^n\|{\mathbf{x}}_i\|^3\rightarrow0,
\end{align}
in probability as $\tau\rightarrow\infty$. From \eqref{eq:43} and \eqref{eq:44},
\begin{equation}\label{eq:55}
  {\breve{{\boldsymbol{\beta}}}}-\hat{{\boldsymbol{\beta}}}_{{{\textnormal{\tiny MLE}}}}=
  -\big(\tilde{{\mathbf{M}}}_X^{\tilde{{\boldsymbol{\beta}}}_0})^{-1}\big)
  \left\{\frac{\dot\ell^*_{\tilde{{\boldsymbol{\beta}}}_0}(\hat{{\boldsymbol{\beta}}}_{{{\textnormal{\tiny MLE}}}})}{n}
    +O_{P|{\mathcal{F}_n}}(\|{\breve{{\boldsymbol{\beta}}}}-\hat{{\boldsymbol{\beta}}}_{{{\textnormal{\tiny MLE}}}}\|^2)\right\}.
\end{equation}
From \eqref{eq:4} of Lemma~\ref{lem2}, $(\tilde{{\mathbf{M}}}_X^{\tilde{{\boldsymbol{\beta}}}_0})^{-1} =O_{P|{\mathcal{F}_n}}(1)$. Combining this with \eqref{eq:29}, \eqref{eq:42} and \eqref{eq:55}
\begin{equation*}
  {\breve{{\boldsymbol{\beta}}}}-\hat{{\boldsymbol{\beta}}}_{{{\textnormal{\tiny MLE}}}}
  =O_{P|{\mathcal{F}_n}}({r}^{-1/2})+o_{P|{\mathcal{F}_n}}(\|{\breve{{\boldsymbol{\beta}}}}-\hat{{\boldsymbol{\beta}}}_{{{\textnormal{\tiny MLE}}}}\|),
\end{equation*}
which implies that
\begin{equation}\label{eq:45}
  {\breve{{\boldsymbol{\beta}}}}-\hat{{\boldsymbol{\beta}}}_{{{\textnormal{\tiny MLE}}}}=O_{P|{\mathcal{F}_n}}({r}^{-1/2}).
\end{equation}

\subsection{Proof of Theorem \ref{thm:6}}
\label{sec:proof-theor-refthm:6}
Denote
\begin{equation}\label{eq:56}
  \frac{\dot\ell^*_{\tilde{{\boldsymbol{\beta}}}_0}(\hat{{\boldsymbol{\beta}}}_{{{\textnormal{\tiny MLE}}}})}{n}
  =\frac{1}{{r}}\sum_{i=1}^{{r}}
  \frac{\{y^*_i-p^*_i(\hat{{\boldsymbol{\beta}}}_{{{\textnormal{\tiny MLE}}}})\}{\mathbf{x}}^*_i}
  {n\pi^*_i(\tilde{{\boldsymbol{\beta}}}_0)}
  \equiv\frac{1}{{r}}\sum_{i=1}^{{r}}{\boldsymbol{\eta}}_i^{\tilde{{\boldsymbol{\beta}}}_0}
\end{equation}
Given ${\mathcal{F}_n}$ and ${\tilde{{\boldsymbol{\beta}}}_0}$, ${\boldsymbol{\eta}}_1^{\tilde{{\boldsymbol{\beta}}}_0}, ..., {\boldsymbol{\eta}}_{{r}}^{\tilde{{\boldsymbol{\beta}}}_0}$ are i.i.d, with mean $\mathbf{0}$ and variance
\begin{align}
  &{\mathrm{Var}}({\boldsymbol{\eta}}_i|{\mathcal{F}_n},\tilde{{\boldsymbol{\beta}}}_0)
    ={r}{\mathbf{V}}_c^{\tilde{{\boldsymbol{\beta}}}_0}
    =\frac{1}{n^2}\sum_{i=1}^n\frac{\{y_i-p_i(\hat{{\boldsymbol{\beta}}}_{{{\textnormal{\tiny MLE}}}})\}^2
    {\mathbf{x}}_i{\mathbf{x}}_i^T}{\pi_i(\tilde{{\boldsymbol{\beta}}}_0)}.\label{eq:46}
\end{align}
Meanwhile, for every $\varepsilon>0$,
\begin{align*}
  & \sum_{i=1}^{{r}} {\mathrm{E}}\{\|{r}^{-1/2}{\boldsymbol{\eta}}_i^{\tilde{{\boldsymbol{\beta}}}_0}\|^2
    I(\|{\boldsymbol{\eta}}_i^{\tilde{{\boldsymbol{\beta}}}_0}\|>{r}^{1/2}\varepsilon)
    |{\mathcal{F}_n},\tilde{{\boldsymbol{\beta}}}_0\}\\
  &\le\frac{1}{{r}^{3/2}\varepsilon}
    \sum_{i=1}^{{r}} {\mathrm{E}}\{\|{\boldsymbol{\eta}}_i^{\tilde{{\boldsymbol{\beta}}}_0}\|^{3}
    I(\|{\boldsymbol{\eta}}_i^{\tilde{{\boldsymbol{\beta}}}_0}\|>{r}^{1/2}\varepsilon)|
    {\mathcal{F}_n},\tilde{{\boldsymbol{\beta}}}_0\}
    \le\frac{1}{{r}^{3/2}\varepsilon}
    \sum_{i=1}^{{r}} {\mathrm{E}}(\|{\boldsymbol{\eta}}_i^{\tilde{{\boldsymbol{\beta}}}_0}\|^{3}|
    {\mathcal{F}_n},\tilde{{\boldsymbol{\beta}}}_0)\\
  &=\frac{1}{{r}^{1/2}}\frac{1}{n^{3}}
    \sum_{i=1}^n\frac{\{y_i-p_i(\hat{{\boldsymbol{\beta}}}_{{{\textnormal{\tiny MLE}}}})\}^{3}
    \|{\mathbf{x}}_i\|^{3}}{\pi_i^{2}(\tilde{{\boldsymbol{\beta}}}_0)}
    \le\frac{1}{{r}^{1/2}}\frac{1}{n^{3}}
    \sum_{i=1}^n\frac{\|{\mathbf{x}}_i\|^{3}}
    {\pi_i^{2}(\tilde{{\boldsymbol{\beta}}}_0)}=o_P(1)
\end{align*}
where the last equality is from Lemma~\ref{lem2}. This and \eqref{eq:46} show that the Lindeberg-Feller conditions are satisfied in probability.
From \eqref{eq:56} and \eqref{eq:46}, by the Lindeberg-Feller central limit theorem \citep[Proposition 2.27 of][]{Vaart:98}, conditionally on ${\mathcal{F}_n}$ and $\tilde{{\boldsymbol{\beta}}}_0$,
\begin{equation*}
  \frac{1}{n}({\mathbf{V}}_c^{\tilde{{\boldsymbol{\beta}}}_0})^{-1/2}\dot\ell^*(\hat{{\boldsymbol{\beta}}}_{{{\textnormal{\tiny MLE}}}})=
  \frac{1}{{r}^{1/2}}\{{\mathrm{Var}}({\boldsymbol{\eta}}_i|{\mathcal{F}_n})\}^{-1/2}\sum_{i=1}^{{r}}{\boldsymbol{\eta}}_i
  \rightarrow N(0,I),
\end{equation*}
in distribution.

Now we exam the distance between ${\mathbf{V}}_c^{\tilde{{\boldsymbol{\beta}}}_0}$ and ${\mathbf{V}}_c$. First,
\begin{align}
  \|{\mathbf{V}}_c-{\mathbf{V}}_c^{\tilde{{\boldsymbol{\beta}}}_0}\|\le\frac{1}{{r}n^2}\sum_{i=1}^n\|{\mathbf{x}}_i\|^2
  \left|\frac{1}{\pi_i}-\frac{1}{\pi_i(\tilde{{\boldsymbol{\beta}}}_0)}\right|
  \label{eq:47}
\end{align}
For the last term in the above equation,
\begin{align}
  &\left|\frac{1}{\pi_i}-\frac{1}{\pi_i(\tilde{{\boldsymbol{\beta}}}_0)}\right|\notag\\
  &\le\left|\frac{\sum_{j=1}^n|y_j-p_j(\hat{{\boldsymbol{\beta}}}_{{{\textnormal{\tiny MLE}}}})|\|{\mathbf{x}}_j\|}
    {|y_i-p_i(\hat{{\boldsymbol{\beta}}}_{{{\textnormal{\tiny MLE}}}})|\|{\mathbf{x}}_i\|}
    -\frac{\sum_{j=1}^n|y_j-p_j(\tilde{{\boldsymbol{\beta}}}_0)|\|{\mathbf{x}}_j\|}
    {|y_i-p_i(\hat{{\boldsymbol{\beta}}}_{{{\textnormal{\tiny MLE}}}})|\|{\mathbf{x}}_i\|}
    \right|\notag\\
  &\quad+\left|\frac{\sum_{j=1}^n|y_j-p_j(\tilde{{\boldsymbol{\beta}}}_0)|\|{\mathbf{x}}_j\|}
    {|y_i-p_i(\hat{{\boldsymbol{\beta}}}_{{{\textnormal{\tiny MLE}}}})|\|{\mathbf{x}}_i\|}
    -\frac{\sum_{j=1}^n|y_j-p_j(\tilde{{\boldsymbol{\beta}}}_0)|\|{\mathbf{x}}_j\|}
    {|y_i-p_i(\tilde{{\boldsymbol{\beta}}}_0)|\|{\mathbf{x}}_i\|}
    \right|\notag\\
  &\le\frac{\sum_{j=1}^n
    |p_j(\tilde{{\boldsymbol{\beta}}}_0)-p_j(\hat{{\boldsymbol{\beta}}}_{{{\textnormal{\tiny MLE}}}})|\|{\mathbf{x}}_j\|}
    {|y_i-p_i(\hat{{\boldsymbol{\beta}}}_{{{\textnormal{\tiny MLE}}}})|\|{\mathbf{x}}_i\|}
    +\left|\frac{1}{|y_i-p_i(\hat{{\boldsymbol{\beta}}}_{{{\textnormal{\tiny MLE}}}})|}
    -\frac{1}{|y_i-p_i(\tilde{{\boldsymbol{\beta}}}_0)|}
    \right|\frac{\sum_{j=1}^n\|{\mathbf{x}}_j\|}{\|{\mathbf{x}}_i\|}\label{eq:48}
\end{align}
Note that
\begin{align}
  |p_j(\tilde{{\boldsymbol{\beta}}}_0)-p_i(\hat{{\boldsymbol{\beta}}}_{{{\textnormal{\tiny MLE}}}})|
  \le\|{\mathbf{x}}_i\|\|\tilde{{\boldsymbol{\beta}}}_0-\hat{{\boldsymbol{\beta}}}_{{{\textnormal{\tiny MLE}}}}\|,\label{eq:49}
\end{align}
and
\begin{align}
  \left|\frac{1}{|y_i-p_i(\hat{{\boldsymbol{\beta}}}_{{{\textnormal{\tiny MLE}}}})|}
    -\frac{1}{|y_i-p_i(\tilde{{\boldsymbol{\beta}}}_0)|}\right|
  &=\left|e^{(2y_i-1){\mathbf{x}}_i^T\hat{{\boldsymbol{\beta}}}_{{{\textnormal{\tiny MLE}}}}}
    -e^{(2y_i-1){\mathbf{x}}_i^T\tilde{{\boldsymbol{\beta}}}_0}\right|\notag\\
  &\le e^{\lambda\|{\mathbf{x}}_i\|}\|{\mathbf{x}}_i\|\|\tilde{{\boldsymbol{\beta}}}_0-\hat{{\boldsymbol{\beta}}}_{{{\textnormal{\tiny MLE}}}}\|.
    \label{eq:50}
\end{align}
From \eqref{eq:47}, \eqref{eq:48}, \eqref{eq:49} and \eqref{eq:50},
\begin{align}\label{eq:51}
  \|{\mathbf{V}}_c-{\mathbf{V}}_c^{\tilde{{\boldsymbol{\beta}}}_0}\|
  \le\frac{\|\tilde{{\boldsymbol{\beta}}}_0-\hat{{\boldsymbol{\beta}}}_{{{\textnormal{\tiny MLE}}}}\|}{{r}}C_1
  =O_{P|{\mathcal{F}_n}}({r}^{-1}r_0^{-1/2}),
\end{align}
where
\begin{align*}
  C_1=\frac{1}{n}\sum_{i=1}^n\frac{\|{\mathbf{x}}_i\|}{|y_i-p_i(\hat{{\boldsymbol{\beta}}}_{{{\textnormal{\tiny MLE}}}})|}
  \frac{1}{n}\sum_{i=1}^n\|{\mathbf{x}}_i\|^2
  +\frac{1}{n}\sum_{i=1}^n\|{\mathbf{x}}_i\|e^{\lambda\|{\mathbf{x}}_i\|}\frac{1}{n}\sum_{i=1}^n\|{\mathbf{x}}_i\|
  =O_P(1).
\end{align*}
From Lemma~\ref{lem3}, \eqref{eq:55} and~\eqref{eq:45},
\begin{equation}\label{eq:57}
  {\breve{{\boldsymbol{\beta}}}}-\hat{{\boldsymbol{\beta}}}_{{{\textnormal{\tiny MLE}}}}=
  -\frac{1}{n}(\tilde{{\mathbf{M}}}_X^{\tilde{{\boldsymbol{\beta}}}_0})^{-1}
  \dot\ell^*_{\tilde{{\boldsymbol{\beta}}}_0}(\hat{{\boldsymbol{\beta}}}_{{{\textnormal{\tiny MLE}}}})+O_{P|{\mathcal{F}_n}}({r}^{-1})
\end{equation}
From \eqref{eq:4} of Lemma~\ref{lem3},
\begin{align}\label{eq:58}
  (\tilde{{\mathbf{M}}}_X^{\tilde{{\boldsymbol{\beta}}}_0})^{-1}-{\mathbf{M}}_X^{-1}
  &=-{\mathbf{M}}_X^{-1}(\tilde{{\mathbf{M}}}_X^{\tilde{{\boldsymbol{\beta}}}_0}-{\mathbf{M}}_X)
    (\tilde{{\mathbf{M}}}_X^{\tilde{{\boldsymbol{\beta}}}_0})^{-1}
  =O_{P|{\mathcal{F}_n}}({r}^{-1/2}).
\end{align}
From \eqref{eq:16}, \eqref{eq:57}, \eqref{eq:51} and \eqref{eq:58},
\begin{align*}
  &{\mathbf{V}}^{-1/2}({\breve{{\boldsymbol{\beta}}}}-\hat{{\boldsymbol{\beta}}}_{{{\textnormal{\tiny MLE}}}})\\
  &=-{\mathbf{V}}^{-1/2}n^{-1}(\tilde{{\mathbf{M}}}_X^{\tilde{{\boldsymbol{\beta}}}_0})^{-1}
    \dot\ell^*(\hat{{\boldsymbol{\beta}}}_{{{\textnormal{\tiny MLE}}}})
    +O_{P|{\mathcal{F}_n}}({r}^{-1/2})\\
  &=-{\mathbf{V}}^{-1/2}{\mathbf{M}}_X^{-1}n^{-1}\dot\ell^*(\hat{{\boldsymbol{\beta}}}_{{{\textnormal{\tiny MLE}}}})
    -{\mathbf{V}}^{-1/2}\{(\tilde{{\mathbf{M}}}_X^{\tilde{{\boldsymbol{\beta}}}_0})^{-1}-{\mathbf{M}}_X^{-1}\}
    n^{-1}\dot\ell^*(\hat{{\boldsymbol{\beta}}}_{{{\textnormal{\tiny MLE}}}})
    +O_{P|{\mathcal{F}_n}}({r}^{-1/2})\\
  &=-{\mathbf{V}}^{-1/2}{\mathbf{M}}_X^{-1}({\mathbf{V}}_c^{\tilde{{\boldsymbol{\beta}}}_0})^{1/2}
    ({\mathbf{V}}_c^{\tilde{{\boldsymbol{\beta}}}_0})^{-1/2}n^{-1}\dot\ell^*(\hat{{\boldsymbol{\beta}}}_{{{\textnormal{\tiny MLE}}}})
    +O_{P|{\mathcal{F}_n}}({r}^{-1/2}).
\end{align*}
The result in Theorem~\ref{thm:1} follows from Slutsky's Theorem\citep[Theorem 6 of][]{Ferguson1996} and the fact that
\begin{align*}
  {\mathbf{V}}^{-1/2}{\mathbf{M}}_X^{-1}({\mathbf{V}}_c^{\tilde{{\boldsymbol{\beta}}}_0})^{1/2}
  ({\mathbf{V}}^{-1/2}{\mathbf{M}}_X^{-1}({\mathbf{V}}_c^{\tilde{{\boldsymbol{\beta}}}_0})^{1/2})^T
  =&{\mathbf{V}}^{-1/2}{\mathbf{M}}_X^{-1}{\mathbf{V}}_c^{\tilde{{\boldsymbol{\beta}}}_0}{\mathbf{M}}_X^{-1}{\mathbf{V}}^{-1/2}\\
  =&{\mathbf{V}}^{-1/2}{\mathbf{M}}_X^{-1}{\mathbf{V}}_c{\mathbf{M}}_X^{-1}{\mathbf{V}}^{-1/2}+O_{P|{\mathcal{F}_n}}(r_0^{-1/2}{r}^{-1/2})\\
  =&\mathbf{I}+O_{P|{\mathcal{F}_n}}(r_0^{-1/2}{r}^{-1/2}),
\end{align*}
which is obtained using \eqref{eq:51}.

\subsection{Proofs for nonrandom covariates}
To prove the theorems for the case of
nonrandom covariates, we need to use the following two assumptions to replace Assumptions~\ref{asp:1} and \ref{asp:4}, respectively.

\begin{assumption}\label{asp:1n}
  As $n\rightarrow\infty$,
  ${\mathbf{M}}_X=n^{-1}\sum_{i=1}^n w_i(\hat{{\boldsymbol{\beta}}}_{{{\textnormal{\tiny MLE}}}}){\mathbf{x}}_i{\mathbf{x}}_i^T$
  goes to a positive-definite matrix in probability and
  $\lim\sup_nn^{-1}\sum_{i=1}^n\|{\mathbf{x}}_i\|^3<\infty$.
\end{assumption}
  \begin{assumption}\label{asp:2n}
    The covariate distribution satisfies that $n^{-1}\sum_{i=1}^n {\mathbf{x}}_i{\mathbf{x}}_i^T$ converges to a positive definite matrix, and
    $\lim\sup_nn^{-1}\sum_{i=1}^n e^{a\|{\mathbf{x}}_i\|}<\infty$ for any
    $a\in\mathbb{R}$.
  \end{assumption}
  Note that $\hat{{\boldsymbol{\beta}}}_{{{\textnormal{\tiny MLE}}}}$ is random, so the condition on
  ${\mathbf{M}}_X$ holds in probability in Assumption~\ref{asp:1n}. $\pi_i$'s could be
  functions of the responses, and the optimal $\pi_i$'s are indeed
  functions of the responses. Thus Assumptions~\ref{asp:2} and
  \ref{asp:3} involve random terms and remain unchanged.
  
  The proof of Lemma~\ref{lem1} does not require the condition that
  $n^{-1}\sum_{i=1}^n\|{\mathbf{x}}_i\|^3=O_P(1)$, so it is automatically valid for
  nonrandom covariates. The proof of Theorem~\ref{thm:1} requires
  $n^{-1}\sum_{i=1}^n\|{\mathbf{x}}_i\|^3=O_P(1)$ in \eqref{eq:53}. If it is replaced
  with $\lim\sup_nn^{-1}\sum_{i=1}^n\|{\mathbf{x}}_i\|^3<\infty$, \eqref{eq:53} still
  holds. Thus Theorem~\ref{thm:1} is valid if Assumptions~\ref{asp:2}
  and \ref{asp:1n} are true.

  Theorem~\ref{thm:2} is built upon Theorem~\ref{thm:1} and does not
  require additional conditions besides Assumption \ref{asp:3}. Thus it is valid under Assumptions~\ref{asp:2},
  \ref{asp:3} and \ref{asp:1n}.

  Theorems~\ref{thm:3} and \ref{thm:4} are proved by the application
  of Cauchy-Schwarz inequality, and they are valid regardless whether
  the covariates are random or nonrandom.

  To prove Theorems~\ref{thm:5} and \ref{thm:6} for nonrandom
  covariates, we first prove Lemma~\ref{lem2}. From Cauchy-Schwarz inequality,
  \begin{align*}
    \frac{1}{n}\sum_{i=1}^n\|{\mathbf{x}}_i\|^{k_1-k_2}e^{\lambda k_2\|{\mathbf{x}}_i\|}
    &\le\bigg\{\bigg(\frac{1}{n}\sum_{i=1}^n\|{\mathbf{x}}_i\|^{2(k_1-k_2)}\bigg)
      \bigg(\frac{1}{n}\sum_{i=1}^n e^{2\lambda k_2\|{\mathbf{x}}_i\|}\bigg)\bigg\}^{1/2}\notag\\
&\le\bigg\{\frac{\{2(k_1-k_2)\}!}{n}\sum_{i=1}^n e^{\|{\mathbf{x}}_i\|}\bigg\}^{1/2}
    \bigg\{\frac{1}{n}\sum_{i=1}^n e^{2\lambda k_2\|{\mathbf{x}}_i\|}\bigg\}^{1/2}
\end{align*}
Thus, under Assumption~\ref{asp:2n},
  \begin{align}
    \lim\sup_n\frac{1}{n}\sum_{i=1}^n\|{\mathbf{x}}_i\|^{k_1-k_2}e^{\lambda k_2\|{\mathbf{x}}_i\|}
    &\le\infty.\label{eq:26}
\end{align}
Combining \eqref{eq:35}, \eqref{eq:36} and \eqref{eq:26}, Lemma~\ref{lem2} follows. With the results in Lemma \ref{lem2}, the proofs of Lemma \ref{lem3} and Theorem~\ref{thm:5}, and Theorem~\ref{thm:6} are the same as those in Section~\ref{sec:proof-theor-refthm:5}, and Section~\ref{sec:proof-theor-refthm:6}, respectively, except that $(n\tau)^{-1}\sum_{i=1}^n\|{\mathbf{x}}_i\|^3\rightarrow0$ deterministically instead of in probability in \eqref{eq:54}. 

\section{Additional numerical results}
\label{sec:addit-numer-results}
In this section, we provide additional numerical results for rare events data and unconditional MSEs. 

\subsection{Further numerical evaluations for rare events data}
\label{sec:furth-numer-eval}

To further investigate the performance of the proposed method for more extreme rare events data, we adopt the model setup with a univariate covariate in \cite{king2001logistic}, namely,
\begin{equation*}
  P(y=1|x)=\frac{1}{1+\exp(-\beta_0-\beta_1x)}.
\end{equation*}
Following \cite{king2001logistic}, we assume that the covariate $x$ follows a standard normal
distribution and consider different values of $\beta_0$ and a fixed
value of $\beta_1=1$. The full data sample size is set to $n=10^6$ and
$\beta_0$ is set to $-7, -9.5, -12.5$, and $-13.5$, generating
responses with the percentages of 1's equaling 0.1493\%, 0.0111\%,
0.0008\%, and 0.0002\% respectively. For the last case there are only
two 1's (0.0002\%) in the full data of $n=10^6$, and this is a very
extreme case of rare events data. For comparison, we also calculate the MSE
of the full data approach using 1000 Bootstrap sample (the gray
dashed line). Results are reported in Figure~\ref{fig:s3}. It is seen
that as the rare event rate gets closer to 0, the performance of the
OSMAC methods relative to the full data Bootstrap gets better. When the rare event rate is 0.0002\%, for the
full data Bootstrap approach, there are 110 cases out of 1000
Bootstrap samples that the MLE are not found, while this occurs for 18,
2, 4, and 1 cases when $r_0=200$, and $r=200, 500, 700$, and $1000$,
respectively.

\begin{figure}[htp]
  \centering
  \begin{subfigure}{0.49\textwidth}
    \includegraphics[width=\textwidth]{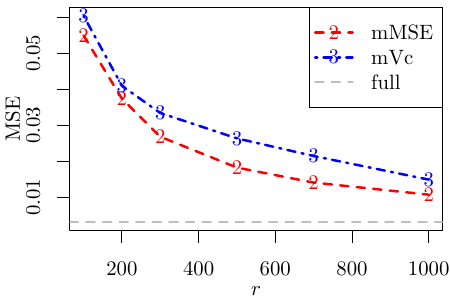}
    \caption{0.1493\% of $y_i$'s are 1}
  \end{subfigure}
  \begin{subfigure}{0.49\textwidth}
    \includegraphics[width=\textwidth]{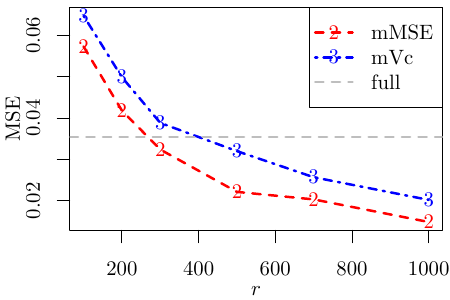}
    \caption{0.0111\% of $y_i$'s are 1}
  \end{subfigure}
  \begin{subfigure}{0.49\textwidth}
    \includegraphics[width=\textwidth]{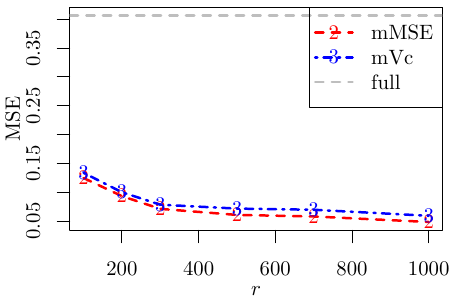}
    \caption{0.0008\% of $y_i$'s are 1}
  \end{subfigure}
  \begin{subfigure}{0.49\textwidth}
    \includegraphics[width=\textwidth]{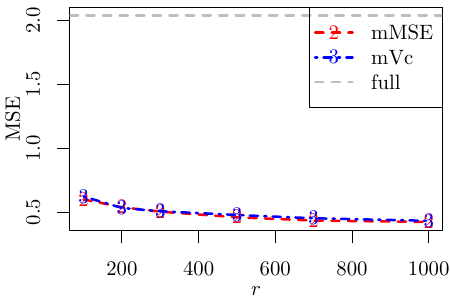}
    \caption{0.0002\% of $y_i$'s are 1}
  \end{subfigure}
  \caption{MSEs for rare events data with different second step subsample size $r$ and a fixed first step subsample size $r_0=200$, where the covariate follows the standard normal distribution.}
  \label{fig:s3}
\end{figure}

\subsection{Numerical results on unconditional MSEs}
\label{sec:numer-results-uncond}
To calculate unconditional MSEs, we generate the full data in each repetition and then apply the subsampling methods. This way, the resultant MSEs are the unconditional MSEs. The exactly same configurations in Section~\ref{sec:numerical-examples} are used. Results are presented in Figure~\ref{fig:s4}. It is seen that the unconditional results are very similar to the conditional results, even for the imbalanced case of nzNormal data sets. For extreme imbalanced data or rare events data, the conditional MSE and the unconditional MSE can be different, as seen in the results in Section~\ref{sec:furth-numer-eval}. 
\begin{figure}[htp]
  \centering
  \begin{subfigure}{0.49\textwidth}
    \includegraphics[width=\textwidth]{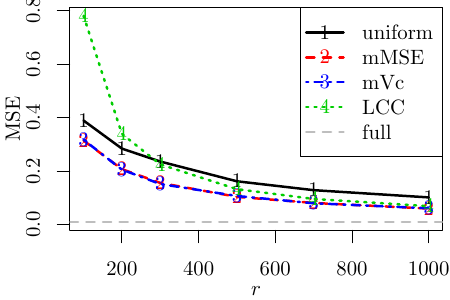}
    \caption{mzNormal}
  \end{subfigure}
  \begin{subfigure}{0.49\textwidth}
    \includegraphics[width=\textwidth]{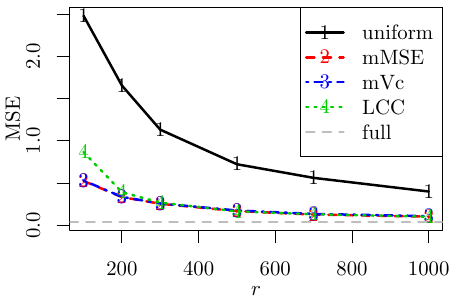}
    \caption{nzNormal}
  \end{subfigure}\\[5mm]
  \begin{subfigure}{0.49\textwidth}
    \includegraphics[width=\textwidth]{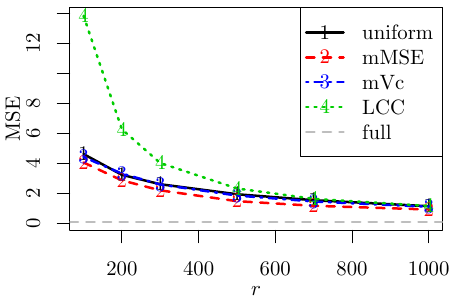}
    \caption{ueNormal}
  \end{subfigure}
  \begin{subfigure}{0.49\textwidth}
    \includegraphics[width=\textwidth]{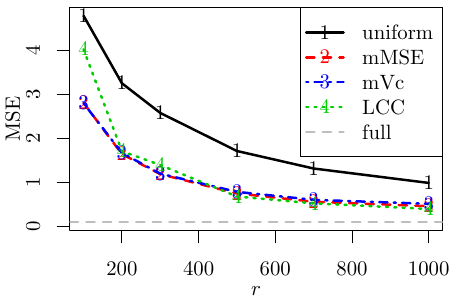}
    \caption{mixNormal}
  \end{subfigure}\\[5mm]
  \begin{subfigure}{0.49\textwidth}
    \includegraphics[width=\textwidth]{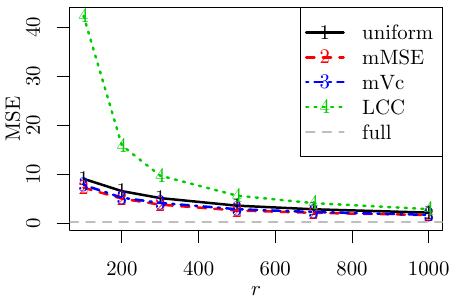}
    \caption{$T_3$}
  \end{subfigure}
  \begin{subfigure}{0.49\textwidth}
    \includegraphics[width=\textwidth]{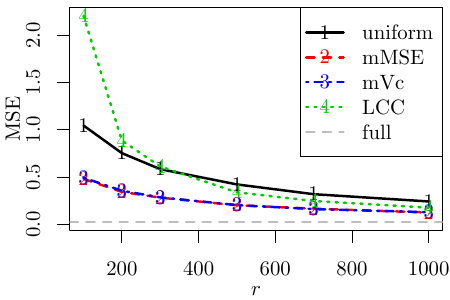}
    \caption{EXP}
  \end{subfigure}\\[5mm]
  \caption{Unconditional MSEs for different second step subsample size ${r}$ with
    the first step subsample size being fixed at $r_0=200$.}
  \label{fig:s4}
\end{figure}

\end{document}